 \documentclass[11pt]{article}
 \usepackage[top=1in,bottom=1in,left=1.5in,right=1.5in]{geometry}
  \usepackage{amsmath,amssymb}
  \usepackage{latexsym}
 \usepackage[dvips]{pstricks} 
  \usepackage{pst-node}
  \usepackage[dvips]{graphicx}
\usepackage{color}
\usepackage{multirow}
\usepackage{multicol}

 \newcommand\R{\mathord{\mathbb R}}
 
 \newcommand\C{{\mathbb C}}

 \newcommand\F{{\mathbb F}}

 \renewcommand\H{\mathord{\mathbb H}}

 \newcommand\N{\mathord{\mathbb N}}
 
 \renewcommand\P{\mathord{\mathbb P}}

 \newcommand\W{\mathord{\mathbb W}}
  
  \renewcommand{\b}{\mathbf{b}\mathnormal}
  \renewcommand{\i}{\mathbf{i}\mathnormal}

  \newcommand{\e}{\mathbf{e}}

  \newcommand{\m}{\mathbf{m}}
  \newcommand{\n}{\mathbf{n}}

  \newcommand{\U}{\mathbf{U}}
  \renewcommand{\u}{\mathbf{u}}
  \renewcommand{\v}{\mathbf{v}}
  \newcommand{\V}{\mathbf{V}}
  \newcommand{\w}{\mathbf{w}}
  \newcommand{\x}{\mathbf{x}}
  
  \newcommand{\y}{\mathbf{y}}
  
  \newcommand{\z}{\mathbf{z}}
  \newcommand{\0}{\mathbf{0}}

 \newcommand\cA{{\cal A}}
 \newcommand\cB{{\cal B}}
 \newcommand\cC{{\cal C}}

 \newcommand\cH{{\cal H}}

 \newcommand\cM{{\cal M}}

 \newcommand\cS{{\cal S}}
 \newcommand\cT{{\cal T}}

 \newcommand\cW{{\cal W}}
 \newcommand\cX{{\cal X}}
 \newcommand\cY{{\cal Y}}
 \newcommand\cZ{{\cal Z}}

 \newcommand\rH{{\rm H}}

 \newcommand\rS{{\rm S}}

  \newcommand{\lan}{\langle}
  \newcommand{\ran}{\rangle}
  \newcommand{\an}[1]{\lan#1\ran}
  \def\diag{\mathop{{\rm diag}}\nolimits}
  \newcommand{\hs}{\hspace*{\parindent}}

  \newcommand{\tr}{\mathop{\mathrm{tr}}\nolimits}

  \newcommand{\trans}{^\top}
  \newcommand{\qed}{\hspace*{\fill} $\Box$\\}

  \newcommand{\dist}{\mathrm{dist}}

  \renewcommand{\rS}{\mathrm{S}}

  \newcommand{\rank}{\mathrm{rank\;}}

  \newtheorem{theo}{\bfseries \hs Theorem}[section]

  \newtheorem{lemma}[theo]{\bfseries \hs Lemma}
  \newtheorem{corol}[theo]{\bfseries \hs Corollary}

  \newtheorem{algo}[theo]{\bfseries \hs Algorithm}
  \newtheorem{example}[theo]{\bfseries \hs Example}

  \numberwithin{equation}{section} 

\newcommand{\bbm}{\begin{bmatrix}}
\newcommand{\ebm}{\end{bmatrix}}

 \newtheorem{problem}[theo]{\bfseries \hs Problem}
 
 \renewcommand{\span}{\mathrm{span}}

 \renewcommand\dim{{\rm dim\;}}

 \newcommand\sym{{\rm sym}}
 
\begin{document}

 \title{Theoretical and computational aspects of entanglement}
 \author{Harm Derksen\footnotemark[1],
  Shmuel Friedland\footnotemark[2], Lek-Heng Lim\footnotemark[3] and Li Wang\footnotemark[2]
 }
 \renewcommand{\thefootnote}{\fnsymbol{footnote}}
  \footnotetext[1]{
 Department of Mathematics, University of Michigan, Ann Arbor, MI 48109, USA
 \texttt{hderksen@umich.edu}}
 \footnotetext[2]{
 Department of Mathematics, Statistics and Computer Science,
 University of Illinois at Chicago, Chicago, Illinois 60607-7045,
 USA, \texttt{friedlan@uic.edu,liwang8@uic.edu}.  }
 \footnotetext[3]{
Computational and Applied Mathematics Initiative, Department of Statistics,
University of Chicago,USA
\texttt{lekheng@galton.uchicago.edu}}

 \renewcommand{\thefootnote}{\arabic{footnote}}
 \date{}
 \maketitle
 \begin{abstract}  We show that the two notions of entanglement: the maximum of the  geometric measure entanglement and the maximum of the nuclear norm is attained for the same states.  We affirm the conjecture of Higuchi-Sudberry on the maximum entangled state of four qubits.  We introduce the notion of $d$-density tensor for mixed $d$-partite states.  We show that $d$-density tensor is separable if and only if its nuclear norm is $1$.  We suggest an alternating method for computing the nuclear norm of tensors.
 We apply the above results to symmetric tensors.

 \end{abstract}

 \noindent \emph{Keywords}: Entanglement, geometric measure of entanglement, spectral and nuclear norms of tensors, symmetric tensors, $d$-qubits, symmetric $d$-qubits,  densty tensors, computation of spectral norm.
 
\noindent\emph{2010 Mathematics Subject Classification}. Primary 15A69, 65K10, 81P40. 
 \section{Introduction}
 The most important notion in quantum mechanics is the notion of (quantum) entaglement of $d$-partite systems \cite{EPR35, Sch35,Sch36}.  (Recall that a $d$-partite state is represented by a $d$-mode tensor $\cT$ of Hilbert-Schmidt norm one: $\|\cT\|=1$.)  A state $\cT$ is called entangled if it is not a product state, (rank one tensor).
 One of the quantitative ways to measure the entanglement of a $d$-partite state $\cT$ is the geometric measure of entanglement of $\cT$ \cite{WG03}.   It is given by the distance of $\cT$ to the variety of product states.  In mathematical terms the geometric measure of entanglement  of $\cT$ is equal to $\sqrt{2\left(1-\|\cT\|_{\infty}\right)}$, where $\|\cT\|_{\infty}$ is the spectral norm of the state $\cT$ \cite{FW16}.   Thus, $\cT$ is entangled if and only if $\|\cT\|_{\infty}<1$.
 
Another important notion in quantum mechanics is a mixed state \cite{Fan57}, which is represented by a hermitian nonnegative definite matrix of trace one.  A pure state is represented by a rank one density matrix.  Mathematically, a $d$-partite quantum state is described by a $2d$-mode tensor $\cH$, which was called density tensor in \cite{FL14, FLn16}.   A mixed density tensor $\cH$ corresponding to the mixture of product states is called separable, or separable state \cite{Per96}.  Thus separable density tensors are generalizations of product states, and inseparable density tensors, i.e., density tensors which are not separable, are analogous to the entangled states. 
The following result for the mixed density tensor was discovered in \cite{FL14}:  Let $\|\cT\|_1$ be the nuclear norm of a $d$-partite tensor, which is the dual norm of the spectral norm \cite{FLn16}.  Then the nuclear norm of a density tensor is at least one, and equality holds if and only if the density is separable.  Hence $\|\cH\|_1$ measures the inseparability of the corresponding mixed states for $d\ge 2$. (For separable bipartite states, i.e., $d=2$, this result was discovered in \cite{Ru02}.)

The aim of this paper to discuss further theoretical and numerical aspect of nuclear norm of tensors initiated in \cite{FL14,FLn16} and their relationship to entanglement.   We first describe our main theoretical results
which are related to entanglement.  With respect to the geometric measure of entanglement, the most entangled state is a $d$-partite state with the minimal spectral norm \cite{TWP09, FW16}.  We propose here another measure of entanglement a $d$-partite state $\cT$: the value of the nuclear norm $\|\cT\|_1$.  Clearly, $\|\cT\|_1\ge 1$ and equality holds if and only if $\cT$ is a product state.  Hence a maximum entangled state with respect to the this measure is a state with maximum $\|\cT\|_1$.  We show that a state has maximum geometric measure of entanglement if and only if it has maximum nuclear norm.  Similarly, the the most inseparable density tensor is the one with the maximum nuclear norm.  We show that the nuclear norm of the most inseparable density tensor is achieved for all pure states which are maximally entangled. 

As pointed out in \cite{GFE09} most qubit states are too entangled to use for quantum computations.  On the other hand, the symmetric $d$-qubits, 
are much less entangled for large values of $d$ \cite{FK16}, and their geometric measure of entanglement is polynomially computable \cite{FW16}.  Furthermore, 
the symmetric qubits are actually available in current designs for quantum computers \cite{AA13}.  Therefore we also discuss in this paper the maximum entangled and
maximum inseparable states corresponding to symmetric tensors, also known as Bosons in physics.

The second part of this paper is devoted to the computational aspects of the nuclear norm.  We first propose a simple numerical algorithm to compute the nuclear norm
of a tensor, which is an analog of the alternating method for computing the spectral norm of a tensor  \cite{LMV00,FMPS,FT15,KB09}.  
Note that this algorithm gives an upper bound on the nuclear norm.  It will usually converge to a local minimum or at least to a critical point.  
We remark that the computation of the spectral  and nuclear norm of tensors is NP-hard (for $d\ge 3$) \cite{HL13,FLn16}.  In general, one would not expect to have a polynomial time algorithm to compute the spectral norm, unless P=NP.

Next we consider the case of symmetric tensors.  We propose a variation of our alternating algorithm to symmetric tensors.  We compare our algorithm to a different approach  suggested by J. Nie in \cite{Ni16}, where the Lasserre hierarchy of semi-definite relaxations based on moments is applied to the non-convex polynomial optimization problem.  The iterations of the Nie's algorithm yield a lower bound for the nuclear norm.

We also try to find most entangled states and most entangled symmetric states using software.  For $3$-qubits our software confirms that the $W$-state is the most entangled state.  For $4$-qubits we prove that the conjectured $4$-qubit given in \cite{HS00} is the most entangled one.
For symmetric $d$-qubits we also confirm numerically that the $d$-symmetric qubits suggested in \cite{AMM10} are the most entangled ones.

We now survey briefly the contents of this paper.  In \S\ref{sec:specnucnorm} we recall the definitions and properties of the spectral and nuclear norms.  We discuss the distortion constants of these two norms with respect to the Hilbert-Schmidt norm. 
In \S\ref{sec:symtens} we discuss similar notions and results for symmetric tensors.
In \S\ref{sec:separab} we discuss the notion of density tensors related to the mixed $d$-partite state.  We show that a density tensor is separable if and only if its nuclear norm is $1$.  In \S\ref{sec:bisymdt} we discuss the density tensors corresponding to the mixed symmetric states, which are called bisymmetric density tensors.
In \S\ref{sec:maxentang} we discuss the notion of the most entangled states and mixed states with respect to spectral and nuclear norms of tensors.  In \S\ref{sec:maxentan34}  we discuss the most entangled $3$ and $4$ qubits.  We show that the $4$-qubit state given by Higuchi-Sudbery \cite{HS00} is the most entangled with respect to the spectral and nuclear norm.  In \S\ref{sec:mineucnrm} we recall two well-known minimization problems: the minimum sum of Euclidean norms \cite{ACCO} and a Second Order Cone Programming (SOCP) \cite{BV04}, which are the foundations for proposing the alternating method for nuclear norm calculation. 
 In \S\ref{sec:altmetnucnrm} and  \S\ref{sec:altmetsymten} we give alternating methods for computing the nuclear norm of nonsymmetric and symmetric tensors respectively. 
In \S\ref{sec:numres} we give some numerical examples to demonstrate the performance of Algorithms \ref{alg:nonsym} and \ref{alg:sym}.

 \section{The spectral and the nuclear norms of tensors}\label{sec:specnucnorm}
 Assume that $d$ is a positive integer and let $\n=(n_1,\ldots,n_d)\in\N^d$ .   In this paper we assume that we are dealing with a field $\F$ which is either is the field of complex numbers $\C$, which is fundamental in quantum  mechanics, or the field of real numbers $\R$, which appears frequently in engineering applications.
Denote by $\F^{\mathbf{n}}$ the $d$-dimensional tensor product $\otimes_{i=1}^d \F^{n_i}$ and by $[d]$ the set of positive integers $\{1,\ldots,d\}$.
Note that the dimension of the vector space $\F^{\n}$ is $N(\n)=\prod_{i=1}^d n_i$.
Recall that $\x=(x_1,\ldots,x_n)\trans\in\F^n, A=[a_{i,j}]\in\F^{m\times n}, \cT=[t_{i_1,\ldots,i_d}]\in\F^{\mathbf{n}}$ are called vector, matrix and $d$-mode tensor 
(for $d\ge 3$), with the entries $x_i, a_{i,j},t_{i_1,\ldots,i_d}$ respectively.  

Assume that $J=\{j_1,\ldots,j_k\}\subseteq[d]$, where
$1\le j_1<\cdots<j_k\le d$. Let $\n'=(n_{j_1},\ldots,n_{j_k})$ and $\cY=[y_{i_{j_1},\ldots,i_{j_k}}]\in \C^{\mathbf{n}'}$. Then
\[\cT\times \cY=\sum_{i_{j_p}\in [n_{j_p}], p\in[k]} t_{i_1,\ldots,i_d}y_{i_{j_1},\ldots,i_{j_k}}\]
is $d-k$ mode tensor.  In particular, for
for $k=d$ one has that the standard inner product on $\F^{\n}$ is given by $\an{\cT,\cY}=\cT\times \bar\cY$.    Then $\|\cT\|=\sqrt{\an{\cT,\cT}}$ is the Hilbert-Schmidt norm of $\cT$.

We now recall the two important norms on $\F^{\n}$, which are of major importance in quantum mechanics for $\F=\C$.
Let
\begin{equation}\label{defPiFn}
\Pi^{\mathbf{n}}(\F)=\{\otimes_{i=1}^d \x_i,\;\x_i\in\F^{n_i},\;\|\x_i\|=1, \; i\in[d]\}.
\end{equation}
$\Pi^{\n}(\F)$, or its projectivization $\P\Pi^{\n}(\F)$, is called the Segre variety.
The spectral norm of a tensor is given by
\begin{equation}\label{defspecnrom}
\|\cT\|_{\infty,\F}=\max\{|\an{\cT,\cX}|, \; \cX\in\Pi^{\n}(\F)\}, \textrm{ for }\cT\in\F^{\n}.
\end{equation}
Clearly, $\|\cT\|_{\infty,\F}\le \|\cT\|$, and for a nonzero tensor $\cT$ the equality $\|\cT\|_{\infty,\F}=\|\cT\|$ if and only if $\cT$ is a rank one tensor.  That is
$\cT=\otimes_{j=1}^d \x_j$, where $\x_j\in \F^{n_j}\setminus\{\0\}$ for $j\in[d]$.
We let $\|\cT\|_{\infty}=\|\cT\|_{\infty,\C}$. It is shown in \cite{FW16} that
\[\|\cT\|_{\infty}=\max\{\Re\left(\an{\cT,\cX}\right), \;\cX\in \Pi^{\n}(\C)\}.\]

The nuclear norm of $\cT\in\F^{\mathbf{n}}$ is defined as follows:
\begin{equation}\label{nucnormdef}
\|\cT\|_{1,\F}=\min\{\sum_{i=1}^r\prod_{j=1}^d \|\x_{i,j}\|, \; \cT=\sum_{i=1}^r \otimes_{j=1}^d \x_{i,j},\; \x_{i,j}\in\F^{n_j}, i\in[r], j\in[d]\}.
\end{equation}
Again we let $\|\cT\|_1=\|\cT\|_{1,\C}$.
It is shown in \cite{FLn16} that we can choose in the characterization \eqref{nucnormdef} $r=N(\n)$ for $\F=\R$.  As $\C^{\n}$ can be viewed as $\R^{\n}\oplus \R^{\n}$ it follows that we can choose in the characterization \eqref{nucnormdef} $r=2N(\n)$ for $\F=\C$.  

Furthermore, it is known that the nuclear norm is the dual norm to
 the spectral norm over $\F$ \cite{FLn16}.   That is,
 \begin{equation}\label{dualcharspecnuc}
 \|\cT\|_{q,\F}=\max\{|\an{\cT,\cY}|, \;\|\cY\|_{p,\F}=1\}, \quad \frac{1}{p}+\frac{1}{q}=1,\;p\in\{1,\infty\}.
 \end{equation}
 Hence the following well known inequality holds
 \begin{equation}\label{prodineq}
 \|\cT\|^2\le \|\cT\|_{1,\F}\|\cT\|_{\infty,\F} \textrm{ for all } \cT\in\F^{\n}.
 \end{equation}
 (Assume for example that $\|\cT\|_{1,\F}=1$ and take in the maximal characterization of $\|\cT\|_{\infty,\F}$ $\cY=\cT$.)
 
 Note that the characterization of spectral norm and the characterization \eqref{dualcharspecnuc} yield that the extreme points of the unit ball of the nuclear norm is the set $\Pi^{\n}(\F)$.  Hence $\|\cT\|_{1,\F}\ge \|\cT\|$.  Equality for a nonzero $\cT$ holds if and only if $\cT$ is a rank
 one tensor.
 
 Observe that by the definition 
 \[ \|\cT\|_{\infty,\R}\le \|\cT\|_{\infty}, \quad \|\cT\|_{1,\R}\ge \|\cT\|_1, \quad \textrm{ for }\cT\in \R^{\n}.\]
 For $d\ge 3$ one may have strict inequalities \cite{FLn16}.  However for matrices, $d=2$, we have always equalities in the above inequalities,
 since the spectral norm and the nuclear norm of a matrix $T$ is the maximal singular value and the sum of singular values respectively. 
 
 Let $\alpha(\n,\F)$ and $\beta(\n,\F)$ be the best constants for comparison of the norms $\|\cT\|_{1,\F}$, $\|\cT\|$ and $\|\cT\|_{\infty,\F}$:
\begin{equation}\label{distortineq}
\frac{1}{\alpha(\n,\F)}\|\cT\|_{1,\F}\le \|\cT\|\le \frac{1}{\beta(\n,\F)}\|\cT\|_{\infty,\F}, \textrm{ for all }\cT\in\F^{\n}.
\end{equation}
Thus
\begin{eqnarray}\label{charalpha}
&&\alpha(\n,\F)=\max\{\|\cT\|_{1,\F}, \;\cT\in\F^{\n}, \|\cT\|=1\},\\  
&&\beta(\n,\F)=\min\{\|\cT\|_{\infty,\F},\;\cT\in\F^{\n},\|\cT\|=1\}.\label{charbeta}
\end{eqnarray}
The following result is well known for matrices.
\begin{lemma}\label{alphbetdim2} Let  $1<m\le n$ be integers.  Let $T\in \F^{m\times n}$ and assume $\|T\|=1$.  Then 
\begin{enumerate}
\item The equality $\|T\|_1=\alpha(m,n)$ holds if and only if $m$ singular values of $T$ are $\frac{1}{\sqrt{m}}$.
\item The equality $\|T\|_\infty=\beta(m,n)$ holds if and only if $m$ singular values of $T$ are $\frac{1}{\sqrt{m}}$.
\end{enumerate}
In particular 
\begin{equation}\label{alphbetd=2}
\alpha(m,n)=\sqrt{m}, \quad \beta(m,n)=\frac{1}{\sqrt{m}}.
\end{equation}
\end{lemma} 
\begin{proof} Let $\sigma_1\ge \cdots\ge\sigma_m\ge 0$.  Then $1=\|T\|^2=\sum_{i=1}^m \sigma_i^2$.  Use Cauchy-Schwarz inequality to deduce that
$\|T\|_1^2=\left(\sum_{i=1}^m \sigma_i\right)^2\le m \left(\sum_{i=1}^m \sigma_i^2\right)=m$.  Equality holds if and only if all singular values of $T$ are $\frac{1}{\sqrt{m}}$.

Observe next that from the equality $1=\sum_{i=1}^m \sigma_i^2$ we deduce that $1\le m\sigma_1^2$.   Hence $\|T\|_{\infty}=\sigma_1\ge \frac{1}{\sqrt{m}}$.  Equality holds if and  only if all singular values of $T$ are $\frac{1}{\sqrt{m}}$.\qed 
\end{proof}

The following theorem generalizes the above lemma to tensors, $d\ge 3$.  Its first part is a simple consequence of the fact that the spectral and the nuclear norms are dual. 
\begin{theo}\label{relalphbet}  Let $d\ge 3$.  Then
\begin{equation}\label{relalphbet1}
\alpha(\n,\F)\beta(\n,\F)=1.
\end{equation}
Assume furthermore that $\cT\in\F^{\n}$ and $\|\cT\|=1$.  If either $\|\cT\|_{1,\F}=\alpha(\n,\F)$ or $\|\cT\|_{\infty,\F}=\beta(\n,\F)$ then
\begin{equation}\label{optcondT}
\|\cT\|_{1,\F}\|\cT\|_{\infty,\F}=\|\cT\|^2=1.  
\end{equation}
That is, $\|\cT\|_{1,\F}=\alpha(\n,\F)$ if and only if $\|\cT\|_{\infty,\F}=\beta(\n,\F)$.
\end{theo}  
\begin{proof}  The dual characterization of $\|\cT\|_{\infty,\F}$ \eqref{dualcharspecnuc} yields
\[\|\cT\|_{\infty,\F}=\max_{\cY\ne 0} \frac{|\an{\cT,\cY}|}{\|\cY\|_{1,\F}}\ge \max_{\cY\ne 0} \frac{|\an{\cT,\cY}|}{\alpha(\n,\F)\|\cY\|}=\frac{1}{\alpha(\n,\F)}\|\cT\|.\]
Hence $\beta(\n,\F)\ge \frac{1}{\alpha(\n,\F)}$.  The maximal characterization of $\|\cT\|_{1,\F}$ yields
\[\|\cT\|_{1,\F}=\max_{\cY\ne 0} \frac{|\an{\cT,\cY}|}{\|\cY\|_{\infty,\F}}\le \max_{\cY\ne 0} \frac{|\an{\cT,\cY}|}{\beta(\n,\F)\|\cY\|}=\frac{1}{\beta(\n,\F)}\|\cT\|.\]
Hence $\alpha(\n,\F)\le \frac{1}{\beta(\n,\F)}$.  This proves \eqref{relalphbet1}.  

We now prove the second part of the theorem.
Assume first the case $\F=\C$.  Let $\cT^{\star}\in\C^{\n}$ satisfy $\|\cT^{\star}\|=1$ and $\|\cT^{\star}\|_1=\alpha(\n,\C)$.  Assume that $\cB\in\C^{\n}$ and $\Re\left(\an{\cB,\cT^\star}\right)=0$.  Set $\cT(\varepsilon)=\cT^\star+\varepsilon \cB$. Here $\varepsilon$ 
is a small real number.  So $\|\cT(\varepsilon)\|=1+O(\varepsilon^2)$.  Assume that $\|\cS\|_{\infty}=1$ and $\an{\cT^\star,\cS}=\|\cT^\star\|_{1}$.  Hence
 \[\Re\left(\an{\frac{1}{\|\cT(\varepsilon)\|}\cT(\varepsilon),\cS}\right)\le \|\frac{1}{\|\cT(\varepsilon)\|}\cT(\varepsilon)\|_{1,\F}\le \alpha(n,\F).\]  
 From the maximality of $\|\cT^\star\|_{1,\F}$ it follows that $\Re\left(\an{\cS,\cB}\right)\le 0$.  By replacing $\cB$ by $-\cB$ we deduce that $\Re\left(\an{\cS,\cB}\right)=0$.  

Consider the hyperplane 
$\Re\left(\an{\cX,\|\cT^{\star}\|_1\cT^{\star}}\right)= \|\cT^{\star}\|_1$.  This hyperplane passes through $\cT^\star$.  Consider the balanced convex set $C:=\{\cX\in\C^{\n}, \|\cX\|_{1}\le \|\cT^\star\|_{1}\}$.  We claim that the hyperplane $\Re\left(\an{\cX,\|\cT^{\star}\|_1\cT^{\star}}\right)= \|\cT^{\star}\|_1$ supports this convex set at $\cT^{\star}$.  If not, there exists $\cB, \Re\left(\an{\cB,\cT^\star}\right)=0$ and $T^\star(\varepsilon)$ is in the interior of $C$ for each small positive $\varepsilon$.  Recall that a supporting hyperplane of $C$ at $\cT^{\star}$ is $\Re\left(\an{\cX,\cS}\right)\le \|\cT^{\star}\|_1$, for some $\cS\in\C^{\n}$, where $\|\cS\|_{\infty}=1$ and $\an{\cT^\star,\cS}=\|\cT^\star\|_{1}$.  As $\cT(\varepsilon)$ is in the interior of $C$ for small enough 
$\varepsilon$ it follows that $\Re\left(\an{\cB,\cS}\right)<0$.  This will contradict the previous observation.  Hence 
$\Re\left(\an{\cX,\|\cT^{\star}\|_1\cT^{\star}}\right)= \|\cT^{\star}\|_1$ is a supporting hyperplane to $C$ at $\cT^{\star}$. Therefore $1=\|\|\cT^{\star}\|_1\cT^{\star}\|_{\infty}=
\|\cT^{\star}\|_1\|\cT^{\star}\|_{\infty}$.

Other cases of the second part of the theorem are proved similarly.
\qed\end{proof}

Let $\n\in\N^d, \n'\in \N^{d'}$.  Then $\F^{\n}\otimes \F^{\n'}=\F^{\m}$, where $\m=(\n,\n')\in \N^{d+d'}$.  In what follows we will need the following lemma.
\begin{lemma}\label{multspecnucnorm}  Let $d,d'\in\N$ and assume that $\n\in\N^d, \n'\in \N^{d'}$.  Suppose that $\cT\in\F^{\n}, \cT'\in\F^{\n'}$.  Then
\begin{equation}\label{multspecnucnorm1}
\|\cT\otimes \cT'\|_{\infty,\F}=\|\cT\|_{\infty,\F}\|\cT'\|_{\infty,\F}, \quad \|\cT\otimes \cT'\|_{1,\F}=\|\cT\|_{1,\F}\|\cT'\|_{1,\F}.
\end{equation}
\end{lemma}
\begin{proof}Let $\m=(\n,\n')$.  
Observe first that $\an{\cT\otimes\cT',\cX\otimes \cX'}=\an{\cT,\cX}\an{\cT',\cX'}$ for $\cX\in\F^{n},\cX'\in\F^{\n'}$.
Clearly $\Pi^{\m}(\F)=\Pi^{\n}(\F)\times \Pi^{\n'}(\F)$.  
Hence the first equality in \eqref{multspecnucnorm1} follows from the definition of $\|\cT\otimes \cT'\|_{\infty,\F}$.
We now prove the second equality. Note that a decomposition of $\cT$ and $\cT'$ to a sum of rank one tensors induces a decomposition of $\cT\otimes\cT'$
to a sum of rank one tensors:
\[\cT=\sum_{i=1}^ r \otimes_{j=1}^d \x_{j,i}, \; \cT'=\sum_{i'=1}^ {r'} \otimes_{j'=1}^d \x_{j',i'}', \;
\cT\otimes \cT'=\sum_{i,i'=1}^{r,r'}\left(\otimes_{j=1}^d \x_{j,i}\right)\otimes\left(\otimes_{j'=1}^d \x_{j',i'}'\right).\]
Clearly, 
\[\sum_{i,i'=1}^{r,r'}\|\left(\otimes_{j=1}^d \x_{j,i}\right)\otimes\left(\otimes_{j'=1}^d \x_{j',i'}'\right)\|=\left(\sum_{i=1}^ r \|\otimes_{j=1}^d \x_{j,i}\|\right)
\left(\sum_{i'=1}^ {r'} \|\otimes_{j'=1}^d \x_{j',i'}'\|\right).\]
Hence the minimal characterization of the nuclear norm \eqref{nucnormdef} yields the inequality $\|\cT\otimes\cT'\|_{1,\F}\le \|\cT\|_{1,\F} \|\cT'\|_{1,\F}\|$.
We now prove the opposite inequality.  Recall that
\[\|\cT\otimes\cT'\|_{1,\F}=\max\{|\an{\cT\otimes\cT',\cZ}|, \; \cZ\in \F^{\m}, \|\cZ\|_{\infty,\F}=1\}.\]
Consider the subset of all $\cZ\in\F^{\m}$ of spectral norm one, of the form $\cX\otimes \cX'$, where $\cX\in \F^{\n}, \cX'\in\F^{\n'}$ and $\|\cX\|_{\infty,\F}=\|\cX'\|_{\infty,\F}=1$.
Hence
\begin{eqnarray*}
&&\|\cT\otimes\cT'\|_{1,\F}\ge \\
&&\left(\max\{|\an{\cT,\cX}|, \|\cX\|_{\infty,\F}=1\}\right) \left(\max\{|\an{\cT',\cX'}|\}, \|\cX'\|_{\infty,\F}=1\right)=
\|\cT\|_{1,\F}\|\cT'\|_{1,\F}.
\end{eqnarray*}

\end{proof}

The value of $\beta(\n,\C)$ , and hence of $\alpha(\n,\C)$ is known for $\n=(2,2,2)$, see \S\ref{sec:maxentan34} .

\section{Symmetric tensors}\label{sec:symtens}
A tensor $\cS=[s_{i_1,\ldots,i_d}]\in\otimes^d\F^n$ is called symmetric if $s_{i_1,\ldots,i_d}=
s_{i_{\omega(1)},\ldots,i_{\omega(d)}}$ for every permutation $\omega:[d]\to[d]$.  
Denote by $\rS^d\F^n\subset \otimes^d\F^n$ the vector space of $d$-mode symmetric tensors on $\F^n$.
It is well known that $\dim \rS^d\F^n={n+d-1\choose d}$ \cite{FK16}.
In what follows we assume that $\cS$ is a symmetric tensor and $d\ge 2$, unless stated otherwise.  A tensor
$\cS\in\rS^d\F^n$ defines a unique homogeneous polynomial of degree $d$ in $n$ variables
\begin{equation}\label{defpolfx}
f(\x)=\cS\times\otimes^d\x=\sum_{0\le j_k\le d,k\in[n], j_1+\cdots +j_n=d} \frac{d!}{j_1!\cdots j_n!} f_{j_1,\ldots,j_n} x_1^{j_1}\cdots x_n^{j_n}.
\end{equation}
Conversely, a homogeneous polynomial $f(\x)$ of degree $d$ in $n$ variables defines a unique symmetric $\cS\in\rS^d\F^n$ by the following relation.
Consider the multiset $\{i_1,\ldots,i_d\}$, where each $i_l\in [n]$.  Let $j_k$ be the number of times the integer $k\in [n]$ appears in the multiset  $\{i_1,\ldots,i_d\}$.
Then $\cS_{i_1,\ldots,i_d}=f_{j_1,\ldots,j_n}$.  Furthermore
\begin{equation}\label{symtenhsnorm}
\|\cS\|^2=\sum_{0\le j_k\le d,k\in[n],j_1+\cdots + j_n=d}\frac{d!}{j_1!\cdots j_n!} |f_{j_1,\ldots,j_n}|^2,
\end{equation}
where $\cS_{i_1,\ldots,i_d}=f_{j_1,\ldots,j_n}$.  

The remarkable result of Banach \cite{Ban38} claims that the spectral norm of a symmetric
tensor can be computed as a maximum on the set of rank one symmetric tensors:
\begin{equation}\label{Banthm}
\|\cS\|_{\sigma,\F}=\max\{|\cS\times \otimes^d \x|, \; \x\in\Pi^{n}(\F), \textrm{ for } \cS\in\rS^d\F^n \}.
\end{equation}
This result was rediscovered several times since 1938.  In quantum information theory (QIT), for the case $\F=\C$, it appeared in \cite{Hubetall09}.  In mathematical literature, for the case $\F=\R$, it appeared in \cite{CHLZ12,Fri13}.  (Observe that  a natural generalization of Banach's theorem to partially symmetric tensors is given in \cite{Fri13}.)

The analog of Banach's theorem for the nuclear norm of symmetric tensors was stated in \cite{FLn16}.   Namely, for $\cS\in \rS^d\F^n$ we have the following minimal characterization
\begin{equation}\label{FLBthm}
\|\cS\|_{1,\F}=\min\{\sum_{i=1}^M \|\x_i\|^d, \;\cS=\sum_{i=1}^M \varepsilon_i\otimes^d \x_i, \;\x_i\in\F^n,\varepsilon_i\in\{1,-1\} \textrm{ for }i\in[r].\}
\end{equation}
We can assume that $\varepsilon_i=1$ unless $\F=\R$ and $d$ is even. 
Furthermore, we can assume that $r={n+d-1\choose d}$ for $\F=\R$   and $r=2{n+d-1\choose d}$ for $\F=\C$.

For symmetric tensors we can improve the inequalities \eqref{distortineq} to 
\[\frac{1}{\alpha'(n,d,\F)}\|\cS\|_{1,\F}\le \|\cS\|\le \frac{1}{\beta'(n,d,\F)}\|\cS\|_{\infty,\F}, \textrm{ for all }\cS\in \rS^d\F^n.\]
Here
\begin{eqnarray}\label{defalphapnd}
&&\alpha'(n,d,\F)=\max\{\|\cS\|_{1,\F}, \cS\in\rS^d\F^n,\|\cS\|=1\}, \\
&&\beta'(n,d,\F)=\min\{\|\cS\|_{\infty,\F}, \cS\in\rS^d\F^n,\|\cS\|=1\}.\label{defbetpnd}.
\end{eqnarray}

We state an analog of Theorem \ref{relalphbet} and Lemma \ref{alphbetdim2}.  The proof of this theorem is similar to Theorem \ref{relalphbet} and Lemma \ref{alphbetdim2}
and we leave it to the reader.
\begin{theo}\label{relalphbetsym}  Let $n,d\ge 2$ be integers.  Then 
\begin{enumerate}
\item $\alpha'(n,d,\F) \beta'(n,d,\F)=1$.
\item Assume that $\cS\in \rS^d\F^n$ and $\|\cS\|=1$.  Then $\|\cS\|_{1,\F}=\alpha'(n,d,\F)$ if and only if $\|\cS\|_{\infty,\F}=\beta'(n,d,\F)$.
\item $\alpha'(n,2,\F)=\sqrt{n},\; \beta(n,2,\F)=\frac{1}{\sqrt{n}}$.  
\item Assume that $S$ is an $n\times n$ complex valued symmetric matrix having Frobenius norm one, $\|S\|=1$.  Then $\|S\|_1=\alpha'(n,2,\F)$ if and only if 
$\sqrt{n}S$ is a unitary matrix.
\end{enumerate}
\end{theo}
\section{Separability and nuclear norm}\label{sec:separab}

In quantum physics, a state is a normailized vector $\u\in\C^n$ of length one: $\|\u\|=1$.  Furthermore, the state $\zeta\u$ is identified with $\u$ for each $\zeta\in\C, |\zeta|=1$. 
Suppose that we have a number of sources that emit states $\u_i\in \C^{n}$, independently,  each with probability $p_i>0$ for $i\in [r]$.  The resulting physical system is called a mixed state.
A standard model of von J. Neumann and L. Landau associates the above mixed state with a density matrix $A=\sum_{i=1}^r p_i\u_i\u_i^*$ \cite{Fan57}.  Denote by $\rH_{n,+,1}\subset \C^{n\times n}$ the convex set of all nonnegative definite hermitian matrices with trace 1.  Thus, a state $\u\in \C^n$ iduces the rank one density matrix $\u\u^*$,
which is also called a pure state. Hence a density matrix is a convex combination of pure states.

A state $\cT\in \C^{\n}$ is called a $d$-partite state.  
A $d$-partite state is called a product state, or unentangled, if 
\[\cT=\otimes_{i=1}^d \x_i, \quad \x_i\in\C^{n_i}, \; \|\x_i\|=1,\; i\in[d].\]
Equivalently, $\cT$ is unentangled if the rank of $\cT$ is one.  It is easy, i.e. polynomially computable, to decide if $\cT$ is unentangled.

We now associate with a mixed state in $\C^{\n}$ the following density matrix.  Let $\m=(m_1,\ldots,m_{2d})=(\n,\n)\in \N^{2d}$.  That is 
$m_{j+d}=m_j=n_j$ for all $j\in[d]$.  We view $\C^{\m}$ as $\C^{\n}\otimes\C^{\n}=\C^{\n\times \n}$.  A mixed state in $\C^{\n}$ is represented by a density matrix in $\rH_{N(\n),+,1}$.  We identify $\rH_{N(\n),+,1}$ with $\rH_{\n,+,1}$.  A density matrix $A\in\rH_{\n,+,1}$ has entries $a_{(i_1,\ldots,i_d),(j_1,\ldots,j_d)}$ where $i_k,j_k\in[n_k]$ for $k\in[d]$.
The hermitian condition is $a_{(j_1,\ldots,j_d),(i_1,\ldots,i_d)}=\overline{a_{(i_1,\ldots,i_d),(j_1,\ldots,j_d)}}$.  Furthermore  $A$ is a nonnegative definite matrix with trace $1$: 
\[\sum_{i_1=\cdots=i_d=1}^{n_1,\ldots,n_d} a_{(i_1,\ldots,i_d),(i_1,\ldots,i_d)}=1.\]

Consider $2d$-mode tensors $\cB=[b_{i_,\ldots,i_{2d}}]\in \C^{\m}$.  Viewing $\cB$ as a matrix over $\C^{\n}$  we define the \textit{trace} as
\begin{equation}\label{deftrac4}
\operatorname{tr}(\cB) :=\sum\nolimits_{i_1=\dots=i_d=1}^{n_1,\dots,n_d}   b_{i_1, \ldots, i_d, i_1, \ldots ,i_d}.
\end{equation}
It is straightforward to see that
\begin{equation}\label{tracedectens}
\operatorname{tr}(\x_1 \otimes \dots \otimes \x_{2d}) =\prod\nolimits_{j=1}^d \x_{j+d}^\mathsf{T} \x_j 
\end{equation}
for every $\x_j,\x_{j+d}\in\mathbb{C}^{n_j}$, $j =1,\dots,d$.

We call $\cB$ a hermitian tensor if $b_{(j_1,\ldots,j_d),(i_1,\ldots,i_d)}=\overline{b_{(i_1,\ldots,i_d),(j_1,\ldots,j_d)}}$ for all indices.
Let $\H^{\n\times \n}\subset \mathbb{C}^{\n\times \n}$ be the real vector subspace of $2d$-hermitian tensors . A hermitian tensor $\cB\in\H^{\n\times \n}$ is nonnegative definite if the corresponding $N(\n)\times N(\n)$ hermitian matrix $B$ with entries $b_{(i_1,\ldots,i_d),(j_1,\ldots,j_d)}$ is nonnegative definite.  The convex set of nonnegative definite
hermitian tensors with trace $1$ are identified with density tensors on $\C^{\n}$, and denoted by $\H^{\n\times \n}_{+,1}$.  
With the density matrix $A$ as above we associate the density tensor $\cA=[a_{i_1,\ldots,i_{2d}}]\in\C^{\m}$.  

A state $\cT\in\C^{\n}$ induces the density tensor $\cA=\cT\otimes \overline{\cT}$, which we also call pure state.   
A density tensor $\cA$ is a convex combination of pure states.
A product  state $\cT=\otimes_{i=1}^d \x_i$ induces the pure product state $(\otimes_{i=1}^d \x_i)\otimes(\otimes_{i=1}^d \bar\x_i)$, which we also identify with $\otimes_{i=1}^d(\x_i\x_i^*)$, i.e., the tensor product of pure states.
A density tensor corresponding to a mixed state of product states is called {\it separable}.  
That is, $\cA\in \H^{\n\times\n}_{+,1}$ is separable if it is of the form
\begin{eqnarray}\label{defsepdenten}
&&\cA=\sum_{i=1}^r p_i(\otimes_{j=1}^d \x_{j,i})\otimes(\otimes_{j=1}^d\bar\x_{j,i}),\\ 
&&\x_{j,i}\in\C^{n_j}, \x^*_{j,i}\x_{j,i}=1, j\in[d], p_i\ge 0, i\in[r],\sum_{i=1}^r p_i=0.\notag
\end{eqnarray}
We denote by $\H^{\n\times \n}_{sep}\subset \H^{\n\times \n}_{+,1}$
the convex set of separable density tensors.   The following separability criterion was stated an proved in \cite{FL14} (unpublished):
\begin{lemma}\label{traceineqlem}
Let $\cA\in \mathbb{C}^{\n\times \n}$. Then
\begin{equation}\label{traceineq}
\lvert \operatorname{tr}( \cA) \rvert \le \|\cA\|_1,
\end{equation}
and equality holds if and only if $\cA=z\cB$ for some $z\in\mathbb{C}$ and $\cB\in \mathbb{H}^{\n\times\n}_{sep}$.
Assume furthermore that $\cA$ is a density tensor.  Then $\|\cA\|_1\ge 1$ and equality holds if and only if $\cA$ is separable.
\end{lemma}
\begin{proof}
Let $\cA=\sum\nolimits_{i=1}^{r} \x_{1,i} \otimes \dots \otimes \x_{2d,i}$, where $\|\cA\|_1=\sum\nolimits_{i=1}^{r}\prod\nolimits_{j=1}^{2d}\|\x_{j,i}\|>0$.
In view of \eqref{tracedectens},
$\operatorname{tr}(\cA)= \sum\nolimits_{i=1}^r\prod\nolimits_{j=1}^d(\x_{j+d,i}^\mathsf{T} \x_{j,i} )$.
The Cauchy--Schwarz inequality yields that $|\x_{j+d,i}^\mathsf{T} \|\x_{j,i}|\le \x_{j,i}\| \|\x_{j+d,i}\|$.
Equality holds if and only if $\x_{j+d,i}=z_{j,i}\bar \x_{j,i}$ for some $z_{j,i}\in\mathbb{C}$.  Thus
\[
|\operatorname{tr}\cA|\le  \sum\nolimits_{i=1}^r \left| \prod\nolimits_{j=1}^d(\x_{j+d,i}^\mathsf{T} \x_{j,i} ) \right| \le \sum\nolimits_{i=1}^r \prod\nolimits_{j=1}^{2d} \|\x_{j,i}\|=\|\cA\|_1.
\]
This establishes \eqref{traceineq}.  Suppose that equality holds in \eqref{traceineq}.  Then $\x_{1,i} \otimes \dots \otimes \x_{2d,i}$ is of the form $z_{j,i}(\x_{1,i}\x_{1,i}^*) \otimes \dots \otimes (\x_{d,i}\x_{d,i}^*)$.
Observe that
\[
\operatorname{tr}\bigl( z_{j,i}(\x_{1,i}\x_{1,i}^*) \otimes \dots \otimes (\x_{d,i}\x_{d,i}^*)\bigr)=z_{j,i} \prod\nolimits_{j=1}^d \|\x_{j,i}\|^2.
\]
Without loss of generality we may assume  that $\|\x_{j,i}\|=1$
for $j=1,\dots,d$.  Since equality holds in the triangle inequality it follows that all $z_{j,i}$ must have the same arguments.  Hence $A=z B$ where 
\begin{equation}\label{decBsep}
\cB=\sum\nolimits_{i=1}^r t_i (\x_{1,i}\x_{1,i}^*) \otimes \dots \otimes (\x_{d,i}\x_{d,i}^*),
\end{equation}
where
$\x_{j,i}^*\x_{j,i}=1$ for $ j=1,\dots,d$,  and $ \sum\nolimits_{i=1}^r t_i=1$,  $t_i\ge 0$, for $i =1,\dots,r$.

Conversely, suppose $\cB$ is separable.  Hence $\cB$ is of the above form.  Therefore
\[
\|\cB\|_1\le  \sum\nolimits_{i=1}^r t_i \prod\nolimits_{j=1}^d \|\x_{j,i}\|^2=1.
\]
Clearly, $\operatorname{tr}(\cB)=1$.  In view of \eqref{traceineq}, it follows that $\|\cB\|_1=1$.  Hence a decomposition \eqref{decBsep} of $B$ is minimal
with respect to the nuclear norm.

Assume now that $\cA$ is a density tensor.  Then $\tr(\cA)=1$ and \eqref{traceineq} yields that $\|\cA\|_1\ge 1$.  The above arguments show that $\|\cA\|_1=1$ if and only if
$\cA$ is separable.\qed 
 \end{proof}
 
For $d=2$ this result is due to \cite{Ru02}.
We will use the following hardness result from \cite{Gu02} (see also \cite{Ga08}).
\begin{theo}[Gurvits]
Deciding whether a given density tensor is bipartite separable ($d=2$), is an NP-hard problem.
\end{theo}
From this and Lemma~\ref{traceineqlem}, we immediately deduce the hardness result for tensor nuclear norm \cite{FLn16}:
\begin{corol}
Deciding whether a given $4$-tensor is in the nuclear norm unit ball is an NP-hard problem.
\end{corol}

A simple, (polynomially computable), necessary condition for separability of density tensors, is the positivity of the partial transpose \cite{Per96}.
For bipartite density tensors it is also sufficient if and only if 
$\n=(2,2), \n=(2,3), \n=(3,2)$ \cite{Hor96}.
\section{Bisymmetric density tensors}\label{sec:bisymdt}
Assume that $n_1=\cdots=n_d=n$  and denote $n^{\times d}=(n,\ldots,n)\in\N^d$.  A hermitian tensor $\cA=[a_{i_1,\ldots,i_d,i_{d+1},\ldots,i_{2d}}]\in 
\H^{n^{\times d}\times n^{\times d}}$ is called bisymmetric if it is symmetric with respect to the $d$ indices $(i_1,\ldots,i_d)$ and  $(i_{d+1},\ldots,i_{2d})$ (separately).  Denote by $\H^{n^{\times d}\times n^{\times d}}_{bsym}$ the space of hermitian symmetric tensors.   
We first observe the following result. 
\begin{lemma}\label{specdecsymdenten}  Assume that $n,d\ge 2$ are integers.  Then 
\begin{equation}\label{dimhersymten}
\dim \H^{n^{\times d}\times n^{\times d}}_{bsym}={n+d-1\choose d}^2.
\end{equation}
Furthermore, a hermitian tensor $\cA\in \H^{n^{\times d}\times n^{\times d}}$ is bisymmetric if and only it has the spectral decomposition.
\begin{equation}\label{specdecsymherten}
\cA=\sum_{i=1}^{{n+d-1\choose d}} \lambda_i \cT_i\otimes\overline{\cT}_i, \quad
\cT_i\in \rS^d\C^n, \an{\cT_i,\cT_j}=\delta_{ij}, \lambda_i\in\R, i,j\in[{n+d-1\choose d}].
\end{equation}
\end{lemma}
\begin{proof}  Equality \eqref{dimhersymten} follows from counting the set of indices $(i_1,\ldots,i_d)$, invariant under the action the symmetric group on $[d]$.
(Just consider the indices $(i_1,\ldots,i_d)$ where $1\le i_1\le \cdots\le i_d\le n$.)
  
Clearly, a tensor $\cA$ of the form \eqref{specdecsymherten} is a hermitian bisymmetric tensor.  
Consider a spectral decomposition of a hermitian bisymmetric tensor $\cB$, as a Hermitian matrix.  
\[\cB=\sum_{i=1}^{n^d} \mu_i \cX_i\otimes \overline{\cX}_i, \quad \cX_i\in \C^{n^{\times d}}, \an{\cX_i,\cX_j}=\delta_{ij}, \mu_i\in\R, i,j\in [n^d].\]
Assume that $\mu_i\ne 0$.  Then $\cX_i=\mu_i^{-1}\cB\times \cX_i$. Hence $\cX_i\in \rS^d\C^{n}$. In view of \eqref{dimhersymten} we can have at most
$n+d-1\choose d$ orthonormal vectors in $\rS^d\C^n$.  Therefore $\cB$ has representation \eqref{specdecsymherten}.\qed
\end{proof}

Clearly, the real space of hermitian bisymmetric states is a real subspace of $\rS^d\C^n\otimes\rS^d\C^n$, the subspace of $\otimes^{2d}\C^n$ tensors which are bisymmetric, i.e. symmetric in the first $d$ and the last $d$ inidices.  We first consider the restricition of the spectral and the nuclear norms on $\rS^d\C^n\otimes\rS^d\C^n$.  
\begin{theo}\label{bisymten} Let $\cC\in \rS^d\C^n\otimes\rS^d\C^n$.  Then
\begin{eqnarray}\label{bisymspecnrm}
&&\|\cC\|_{\infty}=\max\{\Re\left(\an{\cC,(\otimes^d\x)\otimes(\otimes^d\y)}\right), \x,\y\in\C^n, \|\x\|=\|\y\|=1\}\\
&&\|\cC\|_1=\min\{\sum_{i=1}^M \|\x_i\|^d\|\y_i\|^d, \;\cC=\sum_{i=1}^M(\otimes^d\x_i)\otimes(\otimes^d\y_i)\}.\label{bisymnucnrm}
\end{eqnarray}
Suppose that $\cA\in \H^{n^{\times d}\times n^{\times d}}_{bsym}$.  Then
\begin{eqnarray}\label{hbisymspecnrm}
\|\cA\|_{\infty}=\max\{\Re\left(\an{\cA,\frac{1}{2}((\otimes^d\x)\otimes(\otimes^d\y)+(\otimes^d\bar\y)\otimes(\otimes^d\bar\x))}\right), \x,\y\in\C^n, \|\x\|=\|\y\|=1\}\\
\|\cA\|_1=\min\{\sum_{i=1}^M \|\x_i\|^d\|\y_i\|^d, \;\cA=\sum_{i=1}^M\frac{1}{2}((\otimes^d\x_i)\otimes(\otimes^d\y_i)+(\otimes^d\bar\y)\otimes(\otimes^d\bar\x))\}.\label{hbisymnucnrm}
\end{eqnarray}
In particular, the set of the extreme points of the restriction of the nuclear norm to $\H^{n^{\times d}\times n^{\times d}}_{bsym}$ is of the form
\begin{equation}\label{extpthbysimnuc}
\frac{1}{2}((\otimes^d\x)\otimes(\otimes^d\y)+(\otimes^d\bar\y)\otimes(\otimes^d\bar\x)), \quad \x,\y\in\C^n, \|\x\|=\|\y\|=1.
\end{equation}
\end{theo}
\begin{proof} The characterization \eqref{bisymspecnrm} follows from Banach's theorem.  The characterization \eqref{bisymnucnrm} follows from the arguments
of the proof of the generalization of the Banach theorem to nuclear norm \cite{FLn16}.  The characterization \eqref{hbisymspecnrm} follows from \eqref{bisymspecnrm} and the the equality
\[ \Re\left(\an{\cA,(\otimes^d\x)\otimes(\otimes^d\y)}\right)=\Re\left(\an{\cA,(\otimes^d\bar\y)\otimes(\otimes^d\bar\x)}\right).\]
The characterization \eqref{hbisymspecnrm} yields that the set of the extreme points of the restriction of the nuclear norm to $\H^{n^{\times d}\times n^{\times d}}_{bsym}$ is given
by \eqref{hbisymnucnrm}.  The arguments in \cite{FLn16} yield the characterization \eqref{hbisymnucnrm}.\qed
\end{proof}

Combine the proof of Lemma \ref{traceineqlem} with \eqref{hbisymnucnrm} to deduce
\begin{corol}\label{charsephbisym}  Let $\cA\in \H^{n^{\times d}\times n^{\times d}}_{bsym}$.  Then the following statements are equivalent:
\begin{enumerate}
\item $\cA$ is separable.
\item $\|\cA\|_1=1$.
\item
\begin{equation}\label{rephbissep}
\cA=\sum_{i=1}^r p_i(\otimes^d\x_i)\otimes(\otimes^d\bar\x_i), \quad  \x_i\in\C^n, \|\x_i\|=1, p_i>0, \sum_{i=1}^r p_i=1.
\end{equation}
\end{enumerate}
\end{corol}

Another criterion to check if a bipartite ($d=2$) density tensor is separable is given in \cite{NZ16}.

\section{Maximally entangled states}\label{sec:maxentang}

One of the main notions in quantum physics is the notion of entanglement.  The entanglement of $\cT$
can be measured in many different ways.  One of them that we discuss here is the \emph{geometric measure of entanglement}.  
Let $\Pi^{\mathbf{n}}:=\Pi^{\mathbf{n}}(\C)$ be the space of the product states \eqref{defPiFn}.   Then the geometric measure of entanglement is the distance of a state $\cT$ to $\Pi^{\mathbf{n}}$:
\[\dist(\cT,\Pi^{\mathbf{n}})=\min\{\|\cT-\cY\|,\;\cY\in\Pi^{\mathbf{n}}\}.\]  

As $\|\cT\|=\|\cY\|=1$ it follows that $\dist(\cT,\Pi^{\mathbf{n}}(\F))^2=2(1-\|\cT\|_{\infty,\F}).$  Hence an equivalent notion of the geometric measure of entanglement is \cite{GFE09}
\begin{equation}\label{defetaT}
\eta(\cT):=-\log_2 \|\cT\|_{\infty}^2.
\end{equation}
Thus $\cT$ is entangled if and only if $\eta(\cT)>0$.  

Recall that for a $d$-partite state $\|\cT\|_1\ge \|\cT\|=1$.  Furthermore, $\cT$ is a product state if and only if $\|\cT\|_1=1$.
Hence another way to measure the entanglement of $\cT$ is:
\begin{equation}\label{defomegT}
\omega(\cT):=\log_2 \|\cT\|_{1}^2.
\end{equation}

Theorem \ref{relalphbet} yields:
\begin{corol}\label{maxentstate}
\begin{eqnarray}\label{defomegn}
\omega(\n):=\max\{\omega(\cT), \cT\in\C^{\mathbf{n}},\|\cT\|=1\}=\max\{\eta(\cT), \cT\in\C^{\mathbf{n}},\|\cT\|=1\}\\
=2\log_2\alpha(\n,\C)=-2\log_2\beta(\n,\C).\notag
\end{eqnarray}
Furthermore, the most entangled states have the minimal spectral norm and maximal nuclear norm.
\end{corol}

Similarly, for density tensor $\cA\in\H^{\n\times \n}_{+,1}$ we can define the inseparability measure as $\log_2\|\cA\|_1$.  Lemma \ref{traceineqlem} yields that
$\log_2\|\cA\|_1\ge 0$, and $\cA$ is separable if and only  if $\log_2\|\cA\|_1= 0$ .  Hence maximum inseparable density has the maximum value of $\log_2\|\cA\|_1$.
We now show that any maximum entangled state in $\C^n$ induces a maximum inseparable density tensor:
\begin{lemma}\label{maxinsep}
A density tensor $\cA\in\H^{\n\times \n}_{+,1}$ satisfies inequality 
\begin{equation}\label {maxinsepin}
\log_2\|\cA\|_1\le 2\log_2\alpha(\n,\C).
\end{equation}
Equality holds if $\cA=\cT\otimes\overline{\cT}$ and $\cT$ is maximum entangled.
\end{lemma}
\begin{proof}  Suppose that $\cT\in\C^{\n}$  is a state.  Then $\|\cT\|_1=\|\overline{\cT}\|_1\le \alpha(\n,\C)$.
Hence $\cB=\cT\otimes \overline{\cT}$ is a pure density tensor, and $\|\cB\|_1=\|\cT\|_1^2\le \alpha(\n,\C)^2$.  Therefore $\log_2\|\cB\|\le 2\log_2\alpha(\n,\C)$.
Clearly, $\|\cB\|_1=\alpha(\n,\C)^2$ if $\cT$ is maximum entangled.  Assume that $\cA\in\H^{\n\times \n}_{+,1}$.  The spectral decomposition of $\cA$ gives a decomposition
of $\cA$ as a mixed tensor: 
\[\cA=\sum_{i=1}^r \lambda_1\cT_i\otimes \overline{\cT}_i, \quad \lambda_i>0, i\in [r], \sum_{i=1}^r =1.\]
As the nuclear norm is a convex function it follows that
\[\|\cA\|_1\le \sum_{i=1}^r \lambda_i \|\cT_i\otimes \overline{\cT}_i\|_1=  \sum_{i=1}^r \lambda_i \|\cT_i\|^2\le \alpha(\n,\C)^2\sum_{I=1}^r \lambda_i=\alpha(\n,\C)^2.\]
\qed
\end{proof}

We conjecture that equality in \eqref {maxinsepin} implies that $\cA$ is a pure density state corresponding to maximum entangled state.

\section{Maximum entangled $3$ and $4$ qubits}\label{sec:maxentan34} 
We start with the following simple lemma.
\begin{lemma}\label{lem:alphaineq}
We have
$$
\beta((n_1,n_2,\dots,n_{d+1}),\F)\geq \frac{\beta((n_1,n_2,\dots,n_d),\F)}{\sqrt{n_{d+1}}}
$$
and
$$
\alpha((n_1,n_2,\dots,n_{d+1}),\F)\leq \sqrt{n_{d+1}}\alpha((n_1,n_2,\dots,n_d),\F).
$$
\end{lemma}
\begin{proof}
If ${\mathcal T}$ is a unit tensor of type $(n_1,n_2,\dots,n_{d+1})$ then we can write
$$
{\mathcal T}=\lambda_1{\mathcal T}_1\otimes e_1+\lambda_2{\mathcal T}_2\otimes e_2+\cdots +\lambda_{n_{d+1}}{\mathcal T}_{n_{d+1}}\otimes e_{n_{d+1}}.
$$
where ${\mathcal T}_1,\dots,{\mathcal T}_{n_{d+1}}$ are unit tensors of type $(n_1,\dots,n_d)$ and $\lambda_1,\dots,\lambda_{n_{d+1}}\in \F$ such that
$|\lambda_1|^2+\cdots+|\lambda_{n_{d+1}}|^2=1$. For some $i$ we have $|\lambda_i|\geq \frac{1}{\sqrt{n_{d+1}}}$.
There exist a simple tensor ${\mathcal X}$ with $|\langle {\mathcal T}_i,{\mathcal X}\rangle|\geq\beta((n_1,\dots,n_{d}),\F)$. Now we get
$$
\|{\mathcal T}\|_{\infty,\F}\geq |\langle {\mathcal T},{\mathcal X}\otimes e_i\rangle|=|\lambda_i|\cdot |\langle {\mathcal T}_i,{\mathcal X}\rangle|\geq \frac{1}{\sqrt{n_{d+1}}} \beta((n_1,\dots,n_d),\F).
$$
\qed\end{proof}
In what follows we use Dirac's notation in this section.  Namely, let $\e_1=(1,0)\trans, \e_2=(0,1)\trans$ be the standard orthonormal basis in $\C^2$.
Recall that $\otimes^d\C^2$ has the standard basis $\otimes_{j=1}^d \e_{i_j}$, where $i_1,\ldots,i_d\in[2]$.  Then Dirac's notation is
\[|(i_1-1)\cdots(i_d-1)\rangle=\otimes_{j=1}^d \e_{i_j}, \quad i_1,\ldots,i_j\in[2].\]

Consider the tensor the $W$ 3-tensor \cite{DVC00}
$$
{\mathcal W}=|W\rangle=\frac{|100\rangle+|010\rangle+|001\rangle}{\sqrt{3}}.
$$
By Banach's theorem, the spectral norm is achieved on a symmetric state
$$
\langle S|=(x\langle 0|+y\langle 1|)^{\otimes 3}=(x\langle 0|+y\langle 1|)\otimes(x\langle 0|+y\langle 1|)\otimes(x\langle 0|+y\langle 1|)
$$
with value
$$
|\langle S|W\rangle|=|\sqrt{3}yx^2|=\sqrt{3}|y| |x|^2.
$$
where $x,y\in \C$ with $|x|^2+|y|^2=1$. An easy calculus exercise shows that the maximum is achieved when $|x|=\sqrt{2}/\sqrt{3}$ and $|y|=1/\sqrt{3}$. We obtain
$$
\|{\mathcal W}\|_{\infty,\C}=\sqrt{3}\left(\frac{\sqrt{2}}{\sqrt{3}}\right)^2\frac{1}{\sqrt{3}}=\frac{2}{3}\mbox{ and }
\|{\mathcal W}\|_{1,\C}=\frac{3}{2}.
$$
(See also \cite[\S6]{FLn16}.)  It is shown in  \cite{CXZ10} that $3$-qubit state $\cT\in\otimes^3\C^2$ is the most entangled if and only if it is locally untiary equivalent to $\cW$.  
That is, let $\U(2)\subset \C^{2\times 2}$ be the unitary group of $2 \times 2$ complex valued matrices.
Then orbit of $\cW$ is defined as
\[\textrm{orb}(\cW)=\{\cT, \;\cT=(A_1\otimes A_2\otimes A_3) \cW, A_1,A_2,A_3\in \U(2)\}.\]
Thus, $\cT$ is maximally entangled if and only if $\cT\in$orb$(\cW)$.  In particular,
\begin{equation}\label{valalphbet23}
\beta(2,2,2,\C)=\frac{2}{3}, \quad \alpha(2,2,2,\C)=\frac{3}{2}.
\end{equation}
A rank one decomposition that achieves the nuclear norm of $\cW$ is:
$$
{\mathcal W}=\frac{1}{6\sqrt{3}}\left[\Big(\textstyle  \sqrt{2}\,|0\rangle +|1\rangle\Big)^{\otimes 3}+\zeta^2 \Big(\textstyle \sqrt{2} \,|0\rangle +\zeta |1\rangle\Big)^{\otimes 3}+
\zeta \Big(\textstyle \sqrt{2} \,|0\rangle +\zeta^2|1\rangle\Big)^{\otimes 3}\right].
$$ 

Combine Lemme \ref{lem:alphaineq} with Lemma \ref{alphbetdim2} and \eqref{valalphbet23} to deduce
 \begin{corol}\label{specnrmminineq}
 \begin{equation}\label{specnrmminineq1}
 \beta(n^{\times (d+1)},\F)\ge \frac{1}{\sqrt{n}}\beta(n^{\times d},\F) \quad \textrm{ for } d\ge 2.
 \end{equation}
 In particular
 \begin{eqnarray}\label{specnrmminineq2}
 &&\beta(n^{\times d},\R)\ge n^{\frac{1-d}{2}} \quad d\ge 2,\\
 &&\beta(2^{\times d},\C)\ge \left(\frac{2}{3}\right) 2^{\frac{-(d-3)}{2}}  \quad d\ge 3.\label{specnrmminineq3}
 \end{eqnarray}
 \end{corol}
 We remark that Theorem \ref{maxrealentqub} shows that the inequality \eqref{specnrmminineq2} is sharp for $n=2$ and  each $d\ge 2$.

For $n\geq 1$ and $\lambda\in \C$ with $|\lambda|=1$, we define a tensor
\begin{equation}\label{defTnlmab}
{\mathcal T}_{n,\lambda}=\frac{1}{\sqrt{2}}\left(\lambda \Big(\frac{|0\rangle +i |1\rangle}{\sqrt{2}}\Big)^{\otimes n}+\overline{\lambda} \Big(\frac{|0\rangle -i |1\rangle}{\sqrt{2}}\Big)^{\otimes n}\right)
\end{equation}
and we will use the convention ${\mathcal T}_n={\mathcal T}_{n,1}$. For example, we have
$$
{\mathcal T}_3=\frac{|000\rangle-|110\rangle-|101\rangle-|011\rangle}{2},
$$
$$
{\mathcal T}_4=\frac{|0000\rangle-|1100\rangle-|1010\rangle-|1001\rangle-|0110\rangle-|0101\rangle-|0011\rangle+|1111\rangle}{2\sqrt{2}}
$$
and
$$
{\mathcal T}_{4,-i}=\frac{|1000\rangle+|0100\rangle+|0010\rangle+|0001\rangle-|1110\rangle-|1101\rangle-|1011\rangle-|0111\rangle}{2\sqrt{2}}.
$$
As a complex tensor, the state ${\mathcal T}_{n,\lambda}$ is not much entangled in the sense of the nuclear or spectral norm. 
\begin{lemma}\label{Tnineq}  Let ${\mathcal T}_{n,\lambda}$ be defined as above.  Then
\begin{equation}\label{Tnineq1}
\|{\mathcal T}_{n,\lambda}\|_{\infty,\C}=1/\sqrt{2}, \quad \|{\mathcal T}_{n,\lambda}\|_{1,\C}=\sqrt{2}.
\end{equation}
\end{lemma}
\begin{proof}  Clearly, $\cT_{n,\lambda}$ is symmetric :
\[\cT_{n,\lambda}=\frac{1}{\sqrt{2}}(\otimes^n\u +\otimes^n\bar\u), \quad\u=\frac{1}{\sqrt{2}}(1,\i)\trans.\] 
Furthermore, $\u$ and $\bar\u$ is an orthonormal basis in $\C^2$.  For $d\ge 2$ view the tensor $\cT_{n,\lambda}$ as a matrix $T$ of dimension $2^{d_1}\times 2^{d_2}$,
where $d_1=\lfloor\frac{d}{2}\rfloor, d_2=\lceil\frac{d}{2}\rceil$.  Hence \eqref{defTnlmab} is a singular value decomposition of $T$.  Thus $\sigma_1(T)=\frac{1}{\sqrt{2}}$
and $\|T\|_1=\sqrt{2}$.  Therefore, $\|\cT_{n,\lambda}\|_{\infty}\le \sigma_1(T)$ and $\|\cT_{n,\lambda}\|_1\ge \|T\|_1$.  However, this singular value decomposition are realized by 
left and right singular vectors which are rank one tensors.  Hence \eqref{Tnineq1} holds.\qed
\end{proof}
\begin{theo}\label{maxrealentqub}
For every mixed real  $n$-qubit state ${\mathcal T}$ we have $\|{\mathcal T}\|_{\infty,\R}\geq 2^{(1-n)/2}$ and $\|{\mathcal T}\|_{1,\R}\leq 2^{(n-1)/2}$ and these inequalities are tight when ${\mathcal T}={\mathcal T}_{n,\lambda}$. In particular, we have
$$
\alpha\Big(((\underbrace{2,2,\dots,2}_n),\R\Big)=2^{(n-1)/2}\mbox{ and }\beta\Big((\underbrace{2,2,\dots,2}_n),\R\Big)=2^{(1-n)/2}.
$$
\end{theo}
\begin{proof}
We prove the theorem by induction on $n$. The case $n=1$ is clear. If ${\mathcal T}$ is a $n$-qubit tensor of unit length, then we can write
$$
{\mathcal T}=|0\rangle\otimes {\mathcal S}_0+|0\rangle\otimes {\mathcal S}_1
$$
with $\|{\mathcal S}_0\|^2+\|{\mathcal S}_1\|^2=1$. For some $i\in \{0,1\}$ we have $\|{\mathcal S}_i\|_{\infty,\R}^2\geq \frac{1}{2}$, so we get
$$
\|{\mathcal T}\|_{\infty,\R}\geq \|{\mathcal S}_i\|_{\infty,\R}\geq 2^{(2-n)/2}\|{\mathcal S}_i\|\geq 2^{(1-n)/2}.
$$
To calculate $\|{\mathcal T}_{n,\lambda}\|_{\infty,\R}$, we use Banach's theorem. We have to optimize
$$
\Big\langle (x\langle0|+y\langle 0|)^{\otimes n}, {\mathcal T}_{n,\lambda}\Big\rangle=\sqrt{2}\left|\Re\left(\lambda \Big(\frac{x+iy}{\sqrt{2}}\Big)^n\right)\right|
$$
under the constraint $x^2+y^2=1$. The optimal value clearly is equal to $2^{(1-n)/2}$ which shows that  $\|{\mathcal T}_{n,\lambda}\|_{\infty,\R}=2^{(1-n)/2}$.
It follows now that ${\mathcal T}_{n,\lambda}$ also has an optimal nuclear norm, which must be equal to $2^{(n-1)/2}$.\qed
\end{proof}

\begin{theo}\label{M4maxentg} Let  $\cM_4=|M_4\rangle$ be the state given in is \cite{HS00}:
\begin{eqnarray}\notag
&&\frac{1}{\sqrt{6}}\big(\e_1\otimes\e_1\otimes\e_2\otimes\e_2+\e_2\otimes\e_2\otimes\e_1\otimes\e_1+
\omega(\e_2\otimes\e_1\otimes\e_2\otimes\e_1+\e_1\otimes\e_2\otimes\e_1\otimes\e_2)\\
&&+\omega^2(\e_2\otimes\e_1\otimes\e_1\otimes\e_2+\e_1\otimes\e_2\otimes\e_2\otimes\e_1)\big)=\cM_4, \;\omega=e^{2\pi\i/3}.\label{4qubM4}
\end{eqnarray}
Then $\|\cM_4\|_{\infty}=\frac{\sqrt{2}}{3}$.  Hence  
\begin{equation}\label{beta4qub}
\beta(2^{\times 4},\C)=\frac{\sqrt{2}}{3}, \quad \alpha(2^{\times 4},\C)=\frac{3}{\sqrt{2}},
\end{equation}
\end{theo}
and $\cM_4$ is the most entangled state.

\begin{proof} Let $\phi_3:\U(2)\to \otimes^3\U(2)$ be the diagonal map $A\mapsto \otimes^3 A$.  So $\phi_3(\U(2))$ acts on $3$-qubits.
We claim that the two dimensional space $\W=\span(\cW_0,\cW_1)$ is invariant under the action $\phi_3(\U(2))$.  Consider first a diagonal
unitary matrix: $A=\diag(a,b)$, where $|a|=|b|=1$.  Then 
\[(\otimes^3 A)\cW_0=a b^2\cW_0, \quad (\otimes^3 A)\cW_1=a^2 b\cW_1.\]
Hence $\phi_3(A)\W=\W)$.  It is left to show that $\W$ is invariant under the action of $\phi_3(S\U(2))$,
where $S\U(2)$ is the special unitary group.  Hence, it is enough to show $\W$ is invariant under the action of Lie group of $\phi_3(S\U(2))$.
The generators of the $S\U(2)$ are $i$ times the Pauli matrices:
\[B_1=\left[\begin{array}{cc} i&0\\0&-i\end{array}\right], B_2=\left[\begin{array}{cc} 0&1\\-1&0\end{array}\right], B_3=\left[\begin{array}{cc} 0&i\\i&0\end{array}\right].\]
The generators of $\phi_3(S\U(2)$ are
\[C_j=B_j\otimes I_2\otimes I_2+I_2\otimes B_j\otimes I_2+I_2\otimes I_2\otimes B_j, quad j\in [3].\]
As $\W$ invariant under the action $\phi_3(A)$, where $A$ is a diagonal unitary matrix, we deduce that $\W$ is invariant under the action of $C_1$.
A straightforward calculation shows that $\W$ is invariant under the action of $C_2$ and $C_3$. 
Furthermore, the action of $\phi_3(\rS\U_2)$ on $\W$ is identical to the action of $S\U(2)$ on the two dimensional subspace $\W$, with respect to the orthogonal
basis $\cW_0,\cW_1$.  That is, given a state $(a,b)\trans \in\C^2$, there exists $A$, 
\[A= \left[\begin{array}{cc} a&-\bar b\\b& \bar a\end{array}\right]\]
such that $\phi_2(A)\cW_0=a\cW_0+b\cW_1$.  Hence
\[\|a\cW_0+b\cW_1\|_{\infty}=\|\phi_2(A)\cW_0\|_{\infty}=\frac{2}{3} \quad \textrm{ for } |a|^2+|b|^2=1.\]
Observe that
\[\cM_4=\frac{1}{\sqrt{2}}\left(\cW_0\otimes|0\rangle + \cW_1\otimes|1\rangle\right).\]
Let $\cX=\x\otimes\y\otimes\u\otimes \v=\cY\otimes \v$, where $\|\x\|=\|\y\|=\| \u\|=\|\v\|=1$, be a product state.  Then
\[\sqrt{2}|\an{\cM_4,\cX}|=|\an{a\cW_0+b\cW_1,\cY}|, \quad a=\an{|0\rangle,\v},b=\an{|1\rangle,\v}.\]
Clearly $|a|^2+|b|^2=\|\v\|^2=1$.  Thus if we maximize on all product states $\cY$ and keep $\v$ fixed we get this this maximum is
$\|a\cW_0+b\cW_1\|_{\infty}=\frac{2}{3}$.  This shows that $\|\cM_4\|_{\infty}=\frac{\sqrt{2}}{3}$.  The inequaltity \eqref{specnrmminineq3} for $d=4$
yields that $\beta(2^{\times 4},\C)=\frac{\sqrt{2}}{3}$.  This equality yields the second equality in \eqref{beta4qub}.  Furthermore, $\cM_4$ is the most entangled
$4$-qubit.\qed
\end{proof}
Numerical simulaitons point out that that  \eqref{specnrmminineq3} is not sharp for $d=5$. 
 
Inequality \eqref{specnrmminineq3} yields that 
$$\omega(2^{\times d})\le d-5 +2\log_2 3 \textrm{ for }d\ge 3.$$
The concentration result of \cite{GFE09} claims that most of $d$ quibits. with respect to the corresponding Haar measure,  satisfy the inequality
$$\eta(\cT)\ge d-2\log_2 d-3 \textrm{ for most of } \cT\otimes^d\in\C^2 \textrm{ for } d\gg 1.$$

A Boson is a symmetric state in $\cS\in\rS^d\C^n$, where $\|\cS\|=1$.  The maximum entanglement of Bosons in $\rS^d\C^n$ is $-2\log_2\alpha'(n,d,\C)$.
For $d=3$ the most entangled state $\cW$ is symmetric.
  The most entangled symmetric $4$-qubit is conjectured to be \cite[Example 6.1]{AMM10}
\[\frac{1}{\sqrt{3}}(\otimes^4 \e_1+\frac{1}{\sqrt{2}}(\e_1\otimes\otimes^3\e_2+ \e_2\otimes\e_1\otimes^2\e_2+\otimes^2\e_2\otimes \e_1\otimes\e_2+\otimes^3\e_2\otimes\e_1)).\]
Its spectral norm is $\frac{1}{\sqrt{3}}\approx 0.5774$ \cite{AMM10}.
It is shown in \cite{FK16} that 
\[-2\log_2\alpha'(2^{\times d},\C)\le \log_2(d+1).\]  
(Note that this is totally different from the results of Theorem \ref{maxrealentqub}.)
Furthermore
\[\eta(\cS)\ge  \log_2 d-\log_2\log_2 d-3, \textrm{ for most } \cS\in\rS^d\C^2 \textrm { for } d\gg 1.\]

\section{Preliminaries to computational part of the paper}\label{sec:mineucnrm}
In this section, we recall two well-known minimization problems, which are the foundations for proposing the alternating method for nuclear norm calculation. 

\subsection{Matrix Decomposition}
Let us consider the following matrix decomposition problem:

\begin{problem}\label{auxminprobmat}  Let $m,n\ge 2$ and $Q\ge \min(m,n)$ be given integers.  Assume that $A\in\F^{n\times m}\setminus\{0\}$ is given. Suppose furthermore that there exists a following decomposition of matrix $A$:
\begin{equation}\label{Adecomp}
	A=\sum_{i=1}^{Q} \u_i\v_i\trans, \quad \v_i\ne \0 \textrm{ for } i\in [Q].
\end{equation}
For fixed vectors $\v_i,i\in [Q]$, how to find a \emph{minimal  decomposition} with respect to the absolute norm sums of the components ?
\end{problem}
A minimal decomposition (\ref{Adecomp}) is a solution of  the following minimization problem:
\begin{equation}\label{min:matrix}
\min\limits_{\y_i\in \F^n, i\in [Q]}\sum_{i=1}^{Q} \|\v_i\|\|\y_i\| \quad \text{s.t.} \quad \sum_{i=1}^{Q} \y_i\v_i\trans=A.
\end{equation}
Especially, if the given $v_i$ are unit vectors, we need to solve the following problem:
\begin{equation}\label{min:matrix:2}
\min\limits_{\y_i\in \F^n, i\in [Q]}\sum_{i=1}^{Q} \|\y_i\| \quad \text{s.t.} \quad \sum_{i=1}^{Q} \y_i\v_i\trans=A.
\end{equation}
This is a well known problem called the minimum sum of Euclidean norms \cite{ACCO}, which can be solved efficiently by reformulating it as a Second Order Cone Programming (SOCP) \cite{BV04}. By bringing extra $Q$ variables $t_i$, we reformulate (\ref{min:matrix:2}) as:
\begin{equation}\label{min:matrix:2:socp}
\min\limits_{\y_i\in \F^n, i\in [Q]}\sum_{i=1}^{Q} t_i \quad \text{s.t.}\, \|\y_i\|\leq t_i, ~i\in[Q], \quad \sum_{i=1}^{Q} \y_i\v_i\trans=A.
\end{equation} 

Instead of solving problem (\ref{min:matrix}) with $Q$ vector variables $\y_i,i\in{[Q]}$, we can further reformulate it as an unconstrained minimization problem.
Consider the linear constraint:
\begin{equation}\label{Adecompcond}
\sum_{i=1}^{Q} \y_i\v_i\trans=A,
\end{equation}
which is a non-homogeneous linear system with $nQ$ variables $\y_i$ and $mn$ equations. Let $\w=(w_1,\ldots,w_s)\trans$ be the vector of the free variables of this system.
Hence the general solution of \eqref{Adecompcond} is 
\begin{equation}\label{genforxi}
\y_i=\mathbf{u}_i+B_i\w, \quad i\in[Q],
\end{equation} 
where $\mathbf{u}_i$ is a special solution of the linear system (\ref{Adecompcond}), and $B_i$ is the basis of the null space. Then the minimum objective function $\sum_{i=1}^{N} \|\v_i\|\|\y_i\|$ of problem (\ref{min:matrix}) boils down to 
\begin{equation}\label{mineucsum}
\min_{\w\in\R^s}\sum_{i=1}^Q \| \mathbf{b}_i +A_i\w\|,  \;i\in[N].
\end{equation}
where $\mathbf{b}_i=\|\v_i\|\mathbf{u}_i,\; A_i=\|\v_i\|B_i$.
By bringing in extra variable $w_{s+i}\in \R$ for each $i\in [Q]$, we reformulate problem \eqref{mineucsum} as the following SOCP:
\begin{equation}\label{secordconeprog}
\min\limits_{\w\in\F^s,w_{s+i},i\in[Q]}\sum_{i=1}^N w_{s+i}, \textrm{ subject to } \|\mathbf{b}_i +A_i\w\|\le w_{s+i}, \quad i\in [Q].
\end{equation}

%

\subsection{Second Order Cone Programming}

We claim that each step minimization step in an alternating method for computing the nuclear norm given in \S\ref{sec:altmetnucnrm} is equivalent to the solution
of the following minimum problem in matrix decomposition:

Indeed, consider the minimization problem in \S\ref{sec:altmetnucnrm} for $k=1$.  Then 
\begin{eqnarray*}
	&&n=n_1, m=\prod_{j=2}^d n_j, Q=N(\F), \y_i=\y_{1,i}, \v_i=\otimes_{j=2}^d \y_{j,i},  i\in Q, \\
	&&\phi(\y_{1,1},\ldots,\y_{d,Q})=\sum_{i=1}^{Q} \|\v_i\|\|\y_i\|=\sum_{i=1}^{Q} \|\y_i\|.
\end{eqnarray*}

Let  $N=N(\F)$ and $\y:=(\y_1\trans,\ldots,\y_{N(\F)}\trans)\trans$.  Then \eqref{Adecompcond} is a solvable system of
linear non-homogeneous system of linear equations. 

%

We now restate the minimization problem \eqref{secordconeprog} as SDPT3, MATLAB software for semidefinite-quadratic-linear programming.
With each $\y_i\in \F^n$  we associate a vector $\z_i=(w_i,\y_i\trans)\trans \in K_n$,
where $K_n\subset\R^{n+1}$ is the Lorentzian cone 
\[K_n:=\{\z=(w,\y\trans)\trans\in \R^{n+1}, \;w\ge \|\y\|\}.\]
In the notation of \cite{TTT06} in our problem we have only the variables $\z_i\in K_n$ for $i\in [N]$. The condition \eqref{Adecompcond}, can be restated as the following condition second condition (P) in \cite{TTT06}:
\begin{equation}\label{linearsidecondP}
\sum_{i=1}^N A_i\z_i=\hat A.
\end{equation}
Here $\hat A\in \R^{mn}$ is a vector formed from the matrix $A=[\mathbf{a}_1\;\cdots\;\mathbf{a}_m]\in\R^{n\times m}$ as follows: $\hat A=(\mathbf{a}_1\trans\;\cdots\;\mathbf{a}_m\trans)\trans$.  Then
\begin{eqnarray}\label{defAi}
&&A_i=[v_{1,i} B\trans\;\cdots v_{m,i} B\trans]\trans \in \R^{mn\times (n+1)}, \quad B=[\0\; I_n]\in \R^{n\times (n+1)},\\
&&\v_i=(v_{1,i},\ldots,v_{m,i})\trans, \quad i\in [N].\notag
\end{eqnarray}
The minimizing function is
\begin{equation}\label{minfunc}
\sum_{i=1}^N \|\v_i\|w_i=\sum_{i=1}^N \mathbf{c}_i \trans\z_i, \quad \mathbf{c}_i=(\|\v_i\|,0,\ldots,0)\trans\in \R^{n+1}, i\in[N].
\end{equation}
\section{An alternating method for computing the nuclear norm of nonsymmetric tensors}\label{sec:altmetnucnrm}
Assume that $d\ge 3$ and $n_1,\ldots,n_d\ge 2$ are positive integers.  Let $\cT\in\F^{n_1\times\cdots\times n_d}$ be a given tensor. 
Recall that the unfolding of a tensor $\cT$ in the $j$-th mode is a matrix $T_j\in\F^{n_j\times \frac{n_1\cdots n_d}{n_j}}$.  The entries of $T_{j}$ are indexed by the rows $i_j\in[n_j]$
and the columns by a $d-1$ tuple $\mathbf{k}=(i_1,\ldots,i_{j-1},i_{j+1},\ldots,i_d)$, where $i_p\in[n_p]$ for $p\in [d]\setminus\{j\}$.  Furthermore the entry $(i_j,\bf{k})$ of $T_j$
is $t_{i_1,\ldots,i_d}$.  Let $r_j=r_j(\cT)=\rank T_j$.  If $\rank T_j <n_j$ one can use Gram-Schmidt process to find an orthonormal basis $\mathbf{b}_{1,j}, \ldots, \mathbf{b}_{r_j,j}$ of the columns space of $T_j$, denoted as $\V_j\subseteq \F^{n_j}$.  We assume that if $r_j=n_j$ then $\mathbf{b}_{1,j},\ldots,\mathbf{b}_{n_j,j}$ is the standard orthonormal basis $\e_{1,j},\ldots,\e_{n_j,j}$ in $\F^{n_j}$.
(It is well known that there exists $\cT\in \F^{n_1\times\ldots \times n_d}$ such that $r_j=n_j$, which are most the tensors, if and only if $n_j\le \frac{n_1\cdots n_d}{n_j}$ \cite{FO14}.)  Thus, $\cT\in \otimes_{j=1}^d \V_j$.  Equivalently, $\cT$ has the Tucker representation
\[\cT=\sum_{l_j\in[r_j], j\in [d]} t'_{l_1,\ldots,l_d} \otimes_{j=1}^d \mathbf{b}_{l_j,j},\quad \text{where}~t'_{l_1,\ldots,l_d}\in \F.\]
For simplicity of the exposition we assume that $\rank T_j=n_j$ for $j\in[d]$.  

We now describe the alternating method for computing the nuclear norm of tensor $\cT$. The iteration process is designed over vector variables
$\y_{k,1},\ldots,\y_{k,N}\in \F^{n_k}$ for each $k\in[d]$ in an alternating scheme.  Let 
\begin{equation}\label{defNdF}
 N(n_1,\ldots,n_d,\R)=\prod_{i=1}^d n_i,\quad N(n_1,\ldots,n_d,\C)=2\prod_{i=1}^d n_i.
\end{equation}

{\bf \underline{Extension Step:}} Assume first that we have a decomposition of the tensor $\cT$ as a sum of rank one (nonzero) tensors
\[\cT=\sum_{i=1}^{N'} \otimes_{j=1}^d \x_{j,i}, \quad N'\le N(n_1,\ldots,n_d,\F).\]
If $N'<N(n_1,\ldots,n_d,\F)$ we first extend the above decomposition  to
\begin{equation}\label{initdectenT}
\cT=\sum_{i=1}^{N} \otimes_{j=1}^d \y_{j,i}, \quad N=N(n_1,\ldots,n_d,\F).
\end{equation}
as follows:  
\begin{enumerate}
\item  $\y_{j,i}=\frac{1}{\|\x_{j,i}\|}\x_{j,i}$ for $j\in[d]\setminus\{k\},i\in [N']$.  
\item $\y_{k,i}=(\prod_{j\in [d]\setminus\{k\}} \|\x_{j,i}\|) \x_{k,i}$ for $i\in[N']$.
\item The vectors $\y_{j,i}\in \F^{n_j}$ for $j\in [d]\setminus\{k\}$ and  $i=N'+1,\ldots, N$ are random norm one vectors.
\item $\y_{k,i}=\0$ for  $i=N'+1,\ldots, N$.
\end{enumerate}

{\bf \underline{Minimization Step:}} We fix the vectors $\y_{j,i}$ for $j\in [d]\setminus\{k\}, i\in[N]$, and view the equality \eqref{initdectenT} as a system of $N(n_1,\ldots,n_d,\F)$
scalar equations in $N$ vector variables $\y_{k,1},\ldots,\y_{k,N}\in\F^{n_k}$.  Define the objective function


\begin{equation}\label{objectfunc:k}
\phi_k(\y_{k,1},\ldots,\y_{k,N})=\sum_{i=1}^N \left(\prod_{j\in[d]\setminus\{k\}} \|\y_{j,i}\|\right) \|\y_{k,i}\|.
\end{equation}

Observe that $\phi_k(\y_{k,1},\ldots,\y_{k,N})$ is an upper bound on $\|\cT\|_{1,\F}$ induced by the decomposition \eqref{initdectenT}.   
Then our minimization problem is
\begin{equation}\label{minaltnucnrm}
\min\left\{\phi_k(\y_{k,1},\ldots,\y_{k,N}), \;\y_{k,1},\ldots,\y_{k,N}\in\F^{n_k} \textrm{ subject to conditions \eqref{initdectenT}}\right\}.
\end{equation}

As   $\prod_{j\in[d]\setminus\{k\}} \|\y_{j,i}\|>0$ for each $i\in [N]$ the function $\phi_k$ is a strict convex function in variables $\y_{k,i}, i\in[N]$.  Hence the above minimum is achieved 
at the unique $\y_{k,1}^{\star},\ldots,\y_{k,N}^\star$.  This gives rise to another decomposition of $\cT$
\begin{equation}\label{objectfuncs}
\cT=\sum_{i=1}^{N} \left(\otimes_{j=1}^{k-1} \y_{j,i}\right)\otimes \y_{k,i}^\star\otimes\left(\otimes_{j=k+1}^d \y_{j,i}\right).
\end{equation}



We now repeat the above {\bf Extension Step} and {\bf Minimization Step} for $k'\in[d]$ until the relative decrease of the objective function $\phi_k$ is smaller than the determined threshold $\epsilon$ (a tiny positive number). We stop the algorithm and output the last value of the target function $\phi_k$ as the nuclear norm of $\cT$ and the corresponding (\ref{objectfuncs}) is a decomposition of tensor $\cT$. 

\begin{algo}\label{alg:nonsym} \emph{(Nonsymmetric Tensor Nuclear Norm Computation)}
	~
	
	{\bf Input:} Nonsymmetric tensor $\cT\in \F^{n_1\times \cdots \times n_d}$, tolerance $\epsilon>0$, iteration $I = 1 $, maximum iteration $I_{\max}$, $N$, and initial point $x_{j,i}\in \F^{n_j}, i\in [N]$,
	let $\phi_{k}^{0} = +\infty $ for all $k\in [d]$.
	
	
	{\bf Step 1:} For $k=1:d$,
	
	(a) do {\bf Extension Step} 1-4 and {\bf Minimization Step} by solving (\ref{minaltnucnrm});
	
	(b) if $|\phi_{k}^{I} - \phi_{k}^{I-1}|<\epsilon$, then break and 
	go to Step 2, otherwise go to Step 3. 
	
	{\bf Step 2:} Output $\|\cT\|_{1,\F} = \phi_{k}^{I}$, and decomposition (\ref{objectfuncs}).
	
	{\bf Step 3:} Set $I = I+1$, if $I\leq I_{\max}$, go to Step 1; otherwise, go to Step 2. 
	
\end{algo}
Note that the value of $\phi$ is a better lower bound for $\|\cT\|_{1,\F}$ then the lower bound $\phi(\y_{1,1},\ldots,\x_{d,N})$ induced by the
decomposition \eqref{initdectenT}.  
Note that it is possible that exactly $N-\hat N$ vectors $\y_{k,i}^\star=\0$.   Hence the decomposition \eqref{objectfuncs} gives rise to decomposition of $\cT$
to $\hat N$ rank one tensors.   
 

\section{An alternating method for computing the nuclear norm of symmetric tensors}\label{sec:altmetsymten}
Denote by $\Sigma_d$ the group of all permutations of $[d]$.  Recall that the cardinality of $\Sigma_d$ is $d!$.
 For each $\x_1,\ldots,\x_d\in\F^n$ let us denote
\[\sym_d(\x_1,\ldots,\x_d)=\frac{1}{d!}\sum_{\sigma\in\Sigma_d}\x_{\sigma(1)}\otimes\cdots\otimes \x_{\sigma(d)}\in \rS^d\F^n.\]
Observe the basic equality
\[\|\x_1\|\cdots\|\x_d\|=\frac{1}{d!}\sum_{\sigma\in\Sigma_d}\|\x_{\sigma(1)}\|\cdots \|\x_{\sigma(d)}\|.\]
Fix $\x_1,\ldots,\x_{d-1}$.  Then $L(\x_1,\ldots,\x_{d-1})$ is a linear operator from $\F^n$ to $\rS^d\F^n$ given by the equality 
\[L(\x_1,\ldots,\x_{d-1})(\x_d)=\sym_d(\x_1,\ldots,\x_d).\]
Observe the equality
\[L(\x_1,\ldots, \x_{k-1},\x_{k+1},\ldots,\x_{d})(\x_k)=\sym_d(\x_1,\ldots,\x_d) \textrm { for each } k\in[d].\]
Suppose we have a decomposition of a symmetric tensor $\cS\in\rS^d\F^n$ to a sum of the rank one tensor
\begin{equation}\label{nonsymdec}
\cS=\sum_{i=1}^{K} \otimes_{j=1}^d \x_{j,i} \in\rS^d\F^n.
\end{equation}
The decomposition \eqref{nonsymdec} induces a symmetric decomposition
\begin{equation}\label{symdec}
\cS=\sum_{i=1}^K L(\x_{1,i},\ldots,\x_{k-1,i},\x_{k+1,i}\ldots\x_{d,i})(\x_{k,i}) \textrm{ for each } k\in[d].
\end{equation}
Note that if $\span(\x_{1,i})=\cdots=\span(\x_{d,i})=\span(\x_{i})\subset\F^n$ for $i=1,\ldots,K$.  Then it follows that $\cS=\sum_{i=1}^K \varepsilon_i\otimes^d\z_i$, where
$\z_i\in \span(\x_i)$ and $\varepsilon_i=\pm 1$ for $i\in[K]$.  (We can always assume that $\varepsilon_i=1$ if $\F=\C$ or $\F=R$ and $d$ is odd.) 
The extension of Banach's theorem for the nuclear
norm of symmetric tensors \cite{FLn16} implies the following decomposition to a sum of rank one symmetric tensors
\begin{equation}\label{minnucdecsymF}
\cS=\sum_{i=1}^K \varepsilon_i\otimes^d\x_i, \x_i\in\F^n,\varepsilon_i=\pm 1,i\in[K], \|\cS\|_{1,\F}=\sum_{i=1}^K \|\x_i\|^d, K\le M(\F).
\end{equation}
Here
\begin{equation}\label{defMF}
M(\R)={n+d-1\choose d-1}, \quad M(\C)=2{n+d-1\choose d-1}.
\end{equation}
Hence for finding the nuclear norm of $\cS\in\rS^d\F^n$ using an alternating method we need consider only the decomposition \eqref{symdec} of $\cS$ of where each $\x_{j,i}\ne \0$
and $K\le M(\F)$.  

Let $r_j(\cS)$ be the rank of the unfolded matrix $T_j(\cS)$ as in \S\ref{sec:altmetnucnrm}.  Clearly, $r_1(\cS)=\cdots=r_d(\cS)=r$.  If $r<n$ then $\V_1=\cdots=\V_r=\V$ is
the columns space of each $T_j(\cS)$.  As in \S\ref{sec:altmetnucnrm} it follows that $\cS\in \rS^d\V$.  In what follows we assume that $r_j(\cS)=n$.

Recall that $\rS^d\F^n$ has a standard orthogonal basis consisting of $\sym(\e_{i_1},\ldots,\ldots,\e_{i_d})$, where $1\le i_1\le i_2\le \cdots\le i_d\le n$ \cite{FK16}. 
This representation gives rise to a representation of the form \eqref{symdec}.   Here $K=M(\R)-\hat K$, where  $\hat K$ is the number of zero coordinates of $\cS=[s_{i_1,\ldots,i_d}]$
satisfying $1\le i_1\le i_2\le\cdots\le i_d\le n$.   As in \S\ref{sec:altmetnucnrm} each representation \eqref{symdec} with $K\le M(\F)$, where all $\x_{i,j}\ne \0$ induces the following representation for a given $k\in[d]$
\begin{equation}\label{repsymSM}
\cS=\sum_{i=1}^{M} L(\y_{1,i},\ldots,\y_{k-1,i},\y_{k+1,i}\ldots\y_{d,i})(\y_{k,i}), \quad M=M(\F),
\end{equation}
which satisfies the following conditions.
\begin{enumerate}
\item $\y_{j,i}=\frac{1}{\|\x_{j,i}\|}\x_{j,i}$ for $j\in[d]\setminus\{k\}$ and $i\in[K]$.
\item $\y_{k,i}=(\prod_{j\in[d]\setminus\{k\}}\|\x_{j,i})\|\x_{k,i}$ for $i\in[K]$.
\item $\y_{j,i}$ is a random vector in $\F^n$ of norm one for $j\in[k]\setminus\{i\}$ and $i=K+1,\ldots,M$.
\item $\y_{k,i}=\0$ for $i=K+1,\ldots,M$.
\end{enumerate}
Then the upper bound for the nuclear norm induced by \eqref{repsymSM} is given by
\begin{equation}\label{defphi}
\psi(\y_{1,1},\ldots,\y_{d,M})=\sum_{i=1}^M \prod_{j=1}^d\|\y_{j,i}\|.
\end{equation}
Not that the above decomposition gives rise to a symmetric decomposition of $\cS$ to a sum of rank one matrix.  The upper bound of the nuclear norm
for this decomposition is also $\phi(\x_{1,1},\ldots,\x_{d,M})$.

Hence the alternating method for computing the nuclear norm of a given symmetric tensor is given by the basic minimum step 
\begin{equation}\label{basminstep}
\min\{\psi(\y_{1,1},\ldots,\y_{d,M}), \textrm{ on } \y_{k,1},\ldots,\y_{k,M}\in\F^n, \textrm{ subject to \eqref{repsymSM}}\}.
\end{equation}
The advantage of this method versus the method in \S\ref{sec:altmetnucnrm} is that we replace $N(\F)$ by much smaller number $M(\F)$.
Furthermore the number of linear conditions is  ${n+d-1\choose d}$ versus $n^d$.
As in the nonsymmetric case, the minimum problem \eqref{basminstep} can be solved by  SDPT3 software.  In the setting \eqref{linearsidecondP},
the vector $\hat A$ has ${n+d-1\choose n-1}$ coordinates.

\begin{algo}\label{alg:sym} \emph{Symmetric Tensor Nuclear Norm Computation} 
	
	{\bf Input:} Symmetric tensor $\cS\in \F^{n\times \cdots \times n}$, tolerance $\epsilon>0$, iteration $I = 1 $, maximum iteration $I_{\max}$, $N$, and initial point $x_{j,i}\in \F^{n_j}, i\in [N]$,
	let $\phi_{k}^{0} = +\infty $ for all $k\in [d]$.
	
	
	{\bf Step 1:} For $k=1:d$,
	
	(a) do {\bf Extension Step} 1-4 and {\bf Minimization Step} by solving (\ref{basminstep});
	
	(b) if $|\phi_{k}^{I} - \phi_{k}^{I-1}|<\epsilon$, then break and 
	go to Step 2, otherwise go to Step 3. 
	
	{\bf Step 2:} Output $\|\cT\|_{1,\F} = \phi_{k}^{I}$, and decomposition (\ref{objectfuncs}).
	
	{\bf Step 3:} Set $I = I+1$, if $I\leq I_{\max}$, go to Step 1; otherwise, go to Step 2. 	
	
\end{algo}

\section{Numerical Experiments}\label{sec:numres}
In this section, we give some numerical examples to demonstrate the performance of Algorithms \ref{alg:nonsym} and \ref{alg:sym}. All the computation are implemented with Matlab R2012a on a MacBook Pro 64-bit OSX (10.9.5) system with 16GB memory and 2.3 GHz Interl Core i7 CPU. The SOCP subproblem is formulated with software Yalmip \cite{yalmip} and solved with SDPT3 \cite{tutuncu2001sdpt3}. We use the default values of the parameters in SDPT3.

In the tables, $K_{\text{rand}}$ stands for the number of random examples;  $\text{Min}_{\F,\text{Iter}}$, $\text{Avg}_{\F,\text{Iter}}$, $\text{Max}_{\F,\text{Iter}}$ stand for the minimum, average, maximum iterations for the $K_{\text{rand}}$ random examples over field $\F = \R$ or $\C$. Similarly, $\text{Min}_{\F,\text{Time}}$, $\text{Avg}_{\F,\text{Time}}$, $\text{Max}_{\F,\text{Time}}$ stand for the minimum, average, maximum computational CPU time for the $K_{\text{rand}}$ random examples over field $\F = \R$ or $\C$. For cleanness of the paper, we keep four digits for all numerical results. 

For different starting point $\x_0$, Algorithms \ref{alg:nonsym} and \ref{alg:sym} might converge to different local minimizer, so we consider to implement Algorithms \ref{alg:nonsym} and \ref{alg:sym} with random starting points $\x_0$ for 30 times, and choose the one with smallest objective function value as the nuclear norm. 

Equation (\ref{optcondT}) is a necessary condition that holds for maximum entangled states. We are interested to find these maximum entangled states over field $\F$, So in the tables, we report both spectral norm and complex norms for each tensor, and we also report the product of these two norms over field $\F$ as {\sf P}$_{\F}$, i.e, $\|\cT\|_{1,\F}\|\cT\|_{\infty,\F}$. 

For each real state, we report its spectral norm and nuclear norm over field $\F = \R$ and $\C$, so there are four norms. For each complex state, we only report its spectral norm and nuclear norm over $\F =\C$.
In the following, we list all the numerical methods used to calculate these norms:
\begin{enumerate}
	\item For nonsymmetric tensor, we calculate its nuclear norm by Algorithm \ref{alg:nonsym}.
	\item For symmetric tensor, we use two numerical methods to calculate its nuclear norm. The first one is Algorithm \ref{alg:sym}, another one is semidefinite relaxation method which is proposed by Nie \cite{Ni16}. We implement semidefinite relaxation method with software Gloptipoly \cite{henrion2009gloptipoly} and the formulated SDP problem is solved by Sedumi \cite{sturm1999using}. 
	\item For real nonsymmetric state, we calculate its real spectral norm by using semidefinite relaxation method \cite{NW14}.
	\item For symmetric state (either real or complex), we calculate its real and complex spectral norms by using the method proposed in \cite{FW16}, which is numerically implemented with Software Bertini \cite{BHSW06} (version 1.5, released in 2015).
	\item For complex nonsymmetric state, we calculate its complex spectral norm by using the semidefinite relaxation method \cite{NW14}, which is originally designed for real spectral norm computation. For complex spectral norm calculation, we can replace each complex variable with two real variables, then (\ref{Banthm}) can be easily reformulated as a homogeneous polynomial optimization problem \cite{NW12}. The method discussed in \cite{NW14} can also be applied to find the complex spectral norm, however, the global optimality certification condition may not always hold.	In this case, the complex spectral norm provided by the semidefinite relaxation method might only be an upper bound. Please refer to \cite{NW14} for details.
\end{enumerate}



 

\subsection{Density Tensor Separability Checking}\label{subsec:density}
In the following, we first test the alternating algorithm \ref{alg:nonsym} on some density tensors whose separability
is known in advance, i.e., its nuclear norm is equal to 1 or not is known. Also, we randomly generate some separable density tensors, and calculate their nuclear norm by Algorithm \ref{alg:nonsym} to see if their nuclear norm is equal to 1 or not. 

\begin{example}\label{cite:exm:2.5}\emph{\cite[Example 2.5]{Nak08}}
	Let us consider the following density tensor $\cT\in \C^{2\times 2\times 2\times 2}$ with $b\in [0,1]$, which is known to be inseparable for any $\frac{1}{3}< b\leq 1$. For $b\in \{1,3/4,2/3,1/2,1/3,1/4,1/5,0\}$, we calculate the nuclear norm by Algorithm \ref{alg:nonsym}, and the results are shown in Table \ref{inseparable:exm:1}. For $b\in [0,1/3]$, we get the numerical nuclear norms of these density tensor are 1, which match Lemma \ref{traceineqlem}. For $b\in (1/3,]$, it is known that tensors $\cT$ are inseparable, and $\text{tr}(\cT) = 1$, so we must have
	$\|\cT\|_{1,F}> 1$	by Lemma \ref{traceineqlem}. For most $b$, we do get the nuclear norm is bigger than 1. However, there is one special case, for $b = 1/2$, we numerically find the nuclear norm of tensor $\cT$ is equal to 1, which contradict to Lemma \ref{traceineqlem}. We conjecture here, the nuclear norm for $b = 1/2$ should be a number that is very close to 1, but numerically, we might not able to detect this fact in our implementation. 
	\begin{equation*}
	\begin{aligned}
	\cT_{1,1,1,1}& = \cT_{2,2,2,2} = \frac{1-b}{4},\ \cT_{1,2,1,2} = \cT_{2,1,2,1}= \frac{1+b}{4},\ \cT_{1,2,2,1} = \cT_{2,1,1,2}=-\frac{b}{2}.
	\end{aligned}
	\end{equation*} 
	\begin{table}[htb]
		\centering
		\begin{scriptsize}
			\begin{tabular}{|c|c|c|c|c|c|c|c|c|c|c|c|}  \hline
				$b$ & 1& 3/4 & 2/3 & 0.60 & 0.55 & 0.52 & 1/2& 1/3& 1/4& 1/5 & 0  \\ \hline
				$\|\cT\|_{1,\mathbb{C}}$	&    2.0000  & 1.5000 
				&  1.3333 & 1.2000 & 1.1000& 1.04000& 1.0000 & 1.0000 & 1.0000 & 1.0000 & 1.0000
				\\ \hline 
			\end{tabular}
		\end{scriptsize}\caption{Nuclear norm of example \ref{cite:exm:2.5}.  } \label{inseparable:exm:1}
	\end{table}     
\end{example}

\begin{example}\label{cite:exm:2.6}\emph{\cite[Example 2.6]{Nak08}} Let us consider the following density tensor $\cT\in \C^{2\times 4\times 2\times 4}$ with parameter $b\in [0,1]$, which is known to be inseparable for $b\in (0,1]$, and separable for $b=0$. In Table \ref{inseparable:exm:2}, we list the nuclear norms for $b= \{1, 3/4, 2/3, 1/2, 1/3, 1/4, 1/5, 0\}$, which are calculated by Algorithm \ref{alg:nonsym}. For $b= 0$, tensor $\cT$ is separable, by Lemma \ref{traceineqlem}, we know its nuclear norm is equal to 1. From Table \ref{inseparable:exm:2}, we can see our numerical result also certifies this fact. For each $b>0$, we get $\|\cT\|_{1,\C}>1$. Since for any $b$, $\text{tr}(\cT) = 1$, by Lemma \ref{traceineqlem}, we have $\|\cT\|_{1,\C}>1$ since $\cT$ is inseparable for any $b>0$. Numerical results in Table \ref{inseparable:exm:2} also certify this fact.
	
	\begin{equation*}
	\begin{aligned}
	\cT_{1,1,1,1}& = \cT_{1,2,1,2}=\cT_{1,3,1,3}=\cT_{1,4,1,4}=\cT_{1,1,2,2} = \cT_{1,2,2,3} =\frac{b}{7b+1},\\
	\cT_{2,2,2,2}& = \cT_{2,3,2,3} = \cT_{1,3,2,4}  = \cT_{2,2,1,1} = \cT_{2,3,1,2} = \cT_{2,4,1,3}= \frac{b}{7b+1},\\
	\cT_{2,1,2,1}& = \cT_{2,4,2,4} = \frac{1+b}{2(7b+1)},\quad \cT_{2,1,2,4}  = \cT_{2,4,2,1} = \frac{\sqrt{1-b^2}}{2(7b+1)},\\ 
	\end{aligned}
	\end{equation*} 
	
	\begin{table}[htb]
		\centering
		\begin{scriptsize}
			\begin{tabular}{|c|c|c|c|c|c|c|c|c|c|c|c|}  \hline 
				$b$ & 1& 3/4 & 2/3 & 1/2& 1/3& 1/4& 1/5 & 0  \\ \hline			   
				$\|\cT\|_{1,\mathbb{C}}$	& 
				1.0106 & 1.0232 & 1.0282 & 1.0367 & 1.0376 & 1.0372 & 1.0362 & 1.0000
				\\ \hline 
			\end{tabular}
		\end{scriptsize}\caption{Nuclear norm of example \ref{cite:exm:2.6}} \label{inseparable:exm:2}
	\end{table}	
\end{example}

\begin{example} \emph{(Random Separable Density Tensor Examples)} We test Algorithm \ref{alg:nonsym} on random density tensors $\cT\in \C^{n_1\times \ldots \times n_d\times n_{d+1} \times \ldots \times n_{2d}}$, which are generated as follows: (i) let $r\in [2\prod_{j=1}^d n_j]$ be a random integer number; (ii) randomly generate $r$ positive numbers $p_i$ satisfy $\sum\limits_{i=1}^rp_i=1$; (iii) randomly generate nonzero vectors $x_{j,i}\in \C^{n_j}, ~j\in [d], i\in [r]$, and normalize them as length 1 vectors; (iv) calculate tensor $\cT$ by formula (\ref{defsepdenten}). Computationally, for each density tensor, we will get its nuclear norm close to 1, with tiny numerical error. In Table \ref{tensor:nuclear:random:density}, we list the average iteration and average computational time. 
	\begin{table}[htb]
		\centering
		\begin{scriptsize}
			\begin{tabular}{|c|c|c|c|c|c|c|c|c|c|} \hline
				$d$ & ($n_1,\ldots,n_d$) & $K_{\text{rand}}$ & Type &    $\text{Avg}_{\text{Iter}}$ & $\text{Avg}_{\text{Time}}$ 
				\\
				\hline 
				2 & (2,4) & 20 & nonsym & 6.10  & 0:00:28     \\ \hline	
				2 & (2,5) & 20 & nonsym &4.45  &  0:00:54 	\\ \hline
				2 & (2,6) & 20 & nonsym &5.55  &  0:03:30	\\ \hline
				2 & (3,4) & 20 & nonsym &5.50 &  0:02:42   	\\ \hline
				2 & (4,4) & 10 & nonsym  & 7.25  &  0:19:22  	\\ \hline 
				3 & (2,2,2) & 20 & nonsym & 11.00 & 0:01:38	\\ \hline	
				3 & (2,2,3) & 20 & nonsym & 10.25 & 0:20:01 	\\ \hline	
		 	3 & (2,3,3) & 5 & nonsym & 7.60 &  3:19:09	\\ \hline
				2 & (2,2) & 20 & sym  & 2  & 0:00:06 \\ \hline
				2 & (3,3) & 20 & sym  & 2  & 0:06:11 \\ \hline
				3 & (2,2,2) & 10 & sym &	3.5 & 0:37:25 \\ \hline 
			\end{tabular}\caption{Computational results for random density tensors} \label{tensor:nuclear:random:density}
		\end{scriptsize}
	\end{table}

\end{example}

\subsection{Nonsymmetric Tensors}\label{subsec:nonsym}
In this subsection, we report the performance of Algorithm \ref{alg:nonsym} on nonsymmetric tensors. We will test Algorithm \ref{alg:nonsym} on the following tensors: (1) explicit nonsymmetric $d$-qubits found from references \cite{Stephen2007,HS00}; (2) random nonsymmetric $d$-qubits; (3) random nonsymmetric tensors. 
 
%

\subsubsection{Nonsymmetric $d$-qubits }
We test Algorithm \ref{alg:nonsym} on some nonsymmetric $d$-qubits
that we can find from references. The tensors and their four norms are reported in Table \ref{tensor:1:small:nonsym:Stephen}. Examples No.1-3 are found from \cite{Stephen2007}. Example No.4 is conjectured in \cite{HS00} as the maximum entangled state for $d=4$ and $\F=\C$.
Example No.4 is a complex state, so its real nuclear norm and spectral norm do not exist, we use ``--'' in the table. We also report the product of the nuclear norm and spectral norm over field $\F$ for each tensor. In Table \ref{tensor:1:small:nonsym:Stephen}, we certify that the necessary condition (\ref{optcondT}) holds for Example No.4 over field $\F = \C$. The equation (\ref{optcondT}) also holds for Examples No. 1-3 over field $\F =\C$, however, their complex nuclear norm is smaller than Example No. 4, which shows that condition (\ref{optcondT}) is only necessary but not sufficient for maximum entangled state.

%
\begin{table}[htb]
	\centering
	\begin{scriptsize}
		\begin{tabular}{|c|c|p{3.9cm}|c|c|c|c|c|c|c|} \hline
			No.  & $d$ & \quad\quad\quad\quad \quad  Tensor  & $\|\cT\|_{1,\R}$  &  $\|\cT\|_{1,\C}$ & $\|\cT\|_{\infty,\R}$ &  $\|\cT\|_{\infty,\C}$ & {\sf P}$_{\R}$ & {\sf P}$_{\C}$ \\ 
			\hline  \multirow{2}{*}{1} & \multirow{2}{*}{4} & $\cT_{1,1,1,1} =\cT_{1,2,2,2} =\frac{1}{2}$ & \multirow{2}{*}{2.0005} & \multirow{2}{*}{2.0002}  & \multirow{2}{*}{0.5000}& \multirow{2}{*}{0.5000} & \multirow{2}{*}{1.0003} & \multirow{2}{*}{1.0001}\\
			& &$\cT_{2,1,1,2}=\cT_{2,2,2,1} =\frac{1}{2}$ &&&&&& \\ \hline 
			  \multirow{4}{*}{2} &\multirow{4}{*}{4}& $\cT_{1,1,1,1} =\cT_{2,2,1,2} =\frac{1}{2}$  & \multirow{4}{*}{2.0002} & \multirow{4}{*}{2.0001} & \multirow{4}{*}{0.5000}   & \multirow{4}{*}{0.5000}  & \multirow{4}{*}{1.0001} & \multirow{4}{*}{1.0001}\\
			&& $\cT_{2,1,2,2} =\cT_{1,2,2,1} =\frac{1}{2\sqrt{2}}$&&&&&&\\
			&& $\cT_{1,1,2,2} = \frac{1}{2\sqrt{2}}$ &&&&&&\\
				&& $\cT_{2,2,2,1} =-\frac{1}{2\sqrt{2}}$ &&&&&&\\  \hline
			\multirow{2}{*}{3} & \multirow{2}{*}{4} & $\cT_{1,1,1,1} =\cT_{1,2,1,2} =\frac{1}{2} $ & \multirow{2}{*}{2.0000} & \multirow{2}{*}{2.0000} & \multirow{2}{*}{0.5000} & \multirow{2}{*}{0.5000} & \multirow{2}{*}{1.0000} & \multirow{2}{*}{1.0000}\\
			&& $\cT_{2,1,2,1}=\cT_{2,2,2,2} =\frac{1}{2}$&&&&&&\\  
			\hline  \multirow{4}{*}{4} &\multirow{4}{*}{4}& $\cT_{1,1,2,2} =\cT_{2,2,1,1} = \frac{1}{\sqrt{6}}$ &\multirow{4}{*}{--} & \multirow{4}{*}{2.1216} & \multirow{4}{*}{--}&  \multirow{4}{*}{0.4714} & \multirow{4}{*}{--} & \multirow{4}{*}{1.0001}\\
			&& $\cT_{2,1,2,1}=\cT_{1,2,1,2} = \frac{\zeta}{\sqrt{6}}$&&&&&&\\ && $\cT_{2,1,1,2}=\cT_{1,2,2,1} =\frac{\zeta^2}{\sqrt{6}}$&&&&&&\\
			&& $\zeta=\frac{-1+i\sqrt{3}}{2}$&&&&& &\\
	 \hline  \multirow{4}{*}{5} & \multirow{4}{*}{5} & $\cT_{2,1,1,1,2} = \cT_{1,2,1,1,1} =\frac{1}{2\sqrt{2}}$ & \multirow{4}{*}{2.8284} & \multirow{4}{*}{2.8281} & \multirow{4}{*}{0.3536}  & \multirow{4}{*}{0.3536} & \multirow{4}{*}{1.0001} & \multirow{4}{*}{1.0000}\\
	 && $ \cT_{1,1,2,1,2}=\cT_{1,1,2,2,1}=\frac{1}{2\sqrt{2}}$&&&&&&\\
	 && $\cT_{2,2,2,1,1}=\cT_{2,2,2,2,2}=\frac{1}{2\sqrt{2}}$&&&& &&\\ 
	 && $\cT_{2,1,1,2,1}=\cT_{1,2,1,2,2} = - \frac{1}{2\sqrt{2}}$&&&& &&\\ 
			\hline
		\multirow{2}{*}{6} & \multirow{2}{*}{6} & 	$\cT_{i,j,k,i,j,k} =\frac{1}{\sqrt{8}}$  &  \multirow{2}{*}{2.8283} & 
		\multirow{2}{*}{2.8283}
		& \multirow{2}{*}{0.3536} & \multirow{2}{*}{0.3536}   
		  & \multirow{2}{*}{1.0001}& \multirow{2}{*}{1.0001} \\
		  && for $i,j,k\in\{1,2\}$ &&&&&&\\
			\hline 
		\end{tabular}
	\end{scriptsize}\caption{Computational results for nonsymmetric $d$-qubits} \label{tensor:1:small:nonsym:Stephen}
\end{table}

%
  
\subsubsection{Random nonsymmetric $d$-qubits}
It is interesting to find the maximum entangled states. In this example, we consider to randomly generate nonsymmetric states, and calculate their nuclear norm by implementing Algorithm \ref{alg:nonsym} over field $\F$. 
We generate a nonsymmetric tensor with each entry being a random variable obeying Gaussian distribution (by {\bf randn} in Matlab), then we normalize the generated random tensor and get a random $d$-qubit $\cT$ with $\|\cT\| = 1$. For $d = 3,4,5,6$, we randomly generate 500 states over field $\F=\R$ and $\C$, and report the maximal nuclear norm we find over these 500 randomly generated states. The computational results are shown in Table \ref{tensor:1:random:nonsym:qubit}. 
The corresponding real states that get the maximal real nuclear norm are shown in Table \ref{tensor:2:random:nonsym:qubit}, and the corresponding complex states that get the maximal complex nuclear norm are shown in Table \ref{tensor:3:random:nonsym:qubit}. 

\begin{table}[htb]
	\centering
	\begin{scriptsize}
		\begin{tabular}{|c|c|c|c|c||c|c|c|c|c|} \hline
			$d$  & $\F$ &  $\|\cT\|_{1,\F}$  &  $\|\cT\|_{\infty,\F}$ &  {\sf P}$_{\F}$ & $d$  & $\F$ &  $\|\cT\|_{1,\F}$  &  $\|\cT\|_{\infty,\F}$ &  {\sf P}$_{\F}$ \\ \hline 
			3  & $\R$ & 1.9743 & 0.5961 & 1.1769  & 3 & $\C$ & 1.4477 & 0.8082 & 1.1700\\ \hline 
			4  & $\R$ & 2.2665 & 0.6062 & 1.3737& 4 & $\C$ & 1.9279 & 0.6187 & 1.1928\\ \hline  
			5  & $\R$ & 2.5323 & 0.5533 & 1.4011 & 5 &  $\C$ & 2.2971 & 0.5213 & 1.1975\\ \hline  
			6  & $\R$ & 3.2200 & 0.4583 & 1.4757 & 6 &  $\C$ & 2.9072 & 0.4502 & 1.3088\\ \hline 	
		\end{tabular}
	\end{scriptsize}\caption{The maximal nuclear norm for randomly 500 nonsymmetric examples} \label{tensor:1:random:nonsym:qubit}
\end{table} 

\begin{table*}[htb]
	\centering
	\begin{scriptsize}
		\begin{tabular}{|c|ll|c|c|c|c|c|}  \hline 
			{\tiny $d = 3$} &\multirow{2}{*}{\tiny	$\cT(:,:,1) = \bbm  -0.3947 & -0.3663 \\ -0.3316  &  0.3077 \ebm $} & \multirow{2}{*}{\tiny $\cT(:,:,2) = \bbm -0.3170 &   0.2405\\ 0.3888   & 0.4426 \ebm$}   \\ 
			{\tiny $\F=\R$} & & \\
			\hline      
	& \multirow{2}{*}{\tiny $\cT(:,:,1,1) = \bbm -0.3363& 0.0504\\
		-0.3620  &   0.2986 \ebm$} & \multirow{2}{*}{\tiny $\cT(:,:,2,1) = \bbm 0.0920&  -0.2712\\
		-0.1341  &   0.2979 \ebm$}\\
	{\tiny $d = 4$} && \\
{\tiny $\F=\R$} &	\multirow{2}{*}{\tiny	$\cT(:,:,1,2) = \bbm -0.3065 & 0.2301\\
		0.2369 &   0.0084 \ebm$} & \multirow{2}{*}{\tiny $\cT(:,:,2,2) = \bbm -0.2535 & -0.2168\\
		-0.0545 &   0.3977 \ebm$} \\ 
	&  & \\\hline  
  &	\multirow{2}{*}{\tiny	$\cT(:,:,1,1,1) = \bbm 0.1924 &  0.2460 \\
	0.1419  &  0.0317 \ebm$}& 	\multirow{2}{*}{\tiny	$\cT(:,:,2,1,1) = \bbm -0.2657 & 0.1729\\
	0.0526  &  0.2609\ebm$} \\
 & &  \\
	&\multirow{2}{*}{\tiny	 $\cT(:,:,1,2,1) = \bbm -0.1521 & 0.0121\\
	0.1114 &  0.2368 \ebm$} & 	\multirow{2}{*}{\tiny	$\cT(:,:,2,2,1) = \bbm 0.1167 & -0.0275 \\
	-0.1487 &  -0.1793 \ebm$} \\
{\tiny $d = 5$} & & \\ 
{\tiny $\F =\R$} & \multirow{2}{*}{\tiny $\cT(:,:,1,1,2) = \bbm 0.2564 & -0.2466 \\
	-0.0789  &  0.1396 \ebm$} & \multirow{2}{*}{\tiny $\cT(:,:,2,1,2) = \bbm 0.2159 & -0.1361\\
	-0.1170  &  0.2367 \ebm$ }\\
 & & \\ 
& \multirow{2}{*}{\tiny $\cT(:,:,1,2,2) = \bbm -0.2018 & 0.1104 \\
	0.2411  &  0.1902 \ebm$}& \multirow{2}{*}{\tiny $\cT(:,:,2,2,2) = \bbm  -0.1590 & -0.2291\\
	-0.0246  &  0.1920  \ebm$} \\ 
  & & \\  \hline
& \multirow{2}{*}{\tiny $\cT(:,:,1,1,1,1) = \bbm  -0.1354 & -0.0111 \\
	0.1972  & -0.1357 \ebm$} & \multirow{2}{*}{\tiny $\cT(:,:,2,1,1,1) = \bbm  
	0.0895 & -0.0218\\
	-0.1806 &  -0.2043 \ebm$} \\
 & & \\  
&\multirow{2}{*}{\tiny  $\cT(:,:,1,2,1,1) = \bbm  0.0984 & -0.0183\\
	0.1156  & -0.0533 \ebm$} & \multirow{2}{*}{\tiny $\cT(:,:,2,2,1,1) = \bbm  
	0.1198 & -0.1609 \\
	0.1198 &   0.0605 \ebm$}\\
 & & \\ 
& \multirow{2}{*}{\tiny  $\cT(:,:,1,1,2,1) = \bbm  -0.1812  & 0.0505\\
	-0.0423 &  -0.0189  \ebm$} & \multirow{2}{*}{\tiny $\cT(:,:,2,1,2,1) = \bbm  
	0.0876 &  0.1185 \\
	0.0177 &  -0.0829  \ebm$}\\
 & & \\ 
  &\multirow{2}{*}{\tiny  $\cT(:,:,1,2,2,1) = \bbm  0.0653  & 0.1180\\
	0.1779 &  0.0927\ebm$} & \multirow{2}{*}{\tiny $\cT(:,:,2,2,2,1) = \bbm 0.0929 & -0.0781\\
	-0.0084 & 0.1328   
	\ebm$} \\ 
{\tiny $d = 6$} & & \\
{\tiny $\F =\R$}  & \multirow{2}{*}{\tiny  $\cT(:,:,1,1,1,2) = \bbm  0.1920  & 0.0142 \\
	-0.0406  &  0.1940\ebm$} & \multirow{2}{*}{\tiny $\cT(:,:,2,1,1,2) = \bbm 0.1931 &-0.1564\\
	0.1086 &  -0.1686    
	\ebm$ }\\ 
 & & \\
& \multirow{2}{*}{\tiny  $\cT(:,:,1,2,1,2) = \bbm  0.1311 & -0.1753\\
	0.0944 &   0.1288 \ebm$} &\multirow{2}{*}{\tiny $ \cT(:,:,2,2,1,2) = \bbm-0.0299 & -0.1418 \\
	0.1816  & -0.2000    
	\ebm$  }\\
 & & \\
& \multirow{2}{*}{\tiny  $\cT(:,:,1,1,2,2) = \bbm  0.1691 &  0.1193\\
	-0.2026  & -0.1685 \ebm$} & \multirow{2}{*}{\tiny $\cT(:,:,2,1,2,2) = \bbm  0.1379 & -0.0011\\
	-0.1126  &  0.0545  
	\ebm$}\\ 
 & & \\  
 & \multirow{2}{*}{\tiny  $\cT(:,:,1,2,2,2) = \bbm  -0.0198 & -0.1335 \\
 	-0.0528  &  0.0693 \ebm$} & \multirow{2}{*}{\tiny $ \cT(:,:,2,2,2,2) =\bbm  0.1190 & 0.0939\\
 		0.1995  &  0.1357    
 		\ebm$} \\   
 	& & \\	\hline
		\end{tabular}
\end{scriptsize}\caption{The most entangled real states for 500 nonsymmetric examples} \label{tensor:2:random:nonsym:qubit}
\end{table*}

\begin{table*}[htb]
	\centering
	\begin{scriptsize}
		\begin{tabular}{|c|ll|c|c|c|c|c|}  \hline 
			{\tiny $d = 3$} &\multirow{2}{*}{\tiny	$\cT(:,:,1) = \bbm  0.3643 + 0.0337i & 0.2663 - 0.3760i\\
				0.2328 - 0.4029i  & 0.2881 - 0.2092i \ebm $} & \multirow{2}{*}{\tiny $\cT(:,:,2) = \bbm 0.1555 - 0.0832i  & -0.2076 + 0.3741i \\
				-0.2736 - 0.0237i &  0.1155 + 0.0880i \ebm$}   \\ 
			{\tiny $\F=\C$} & & \\ 
			\hline       
			& \multirow{2}{*}{\tiny $\cT(:,:,1,1) = \bbm  0.1481 - 0.0069i & 0.2237 - 0.0739i \\
				0.0952 + 0.2873i & -0.1329 + 0.0183i \ebm$} & \multirow{2}{*}{\tiny $\cT(:,:,2,1) = \bbm -0.0470 - 0.1859i & -0.2708 - 0.0454i\\
				-0.1674 + 0.0011i & -0.2443 + 0.0936i  \ebm$}\\ 
			{\tiny $d = 4$} && \\
			{\tiny $\F=\C$} &	\multirow{2}{*}{\tiny	$\cT(:,:,1,2) = \bbm   0.2045 + 0.1013i & -0.0197 - 0.1797i\\
				-0.0932 + 0.2667i &  0.2413 - 0.2267i \ebm$} & \multirow{2}{*}{\tiny $\cT(:,:,2,2) = \bbm -0.1583 + 0.0379i & 0.0913 - 0.2778i \\
				0.2111 + 0.2736i  & 0.2281 + 0.2159i \ebm$} \\ 
			&  & \\\hline  
			&	\multirow{2}{*}{\tiny	$\cT(:,:,1,1,1) = \bbm -0.1892 + 0.0987i & 0.0592 - 0.0647i \\
				-0.1771 + 0.0787i & -0.1299 - 0.1403i  \ebm$}& 	\multirow{2}{*}{\tiny	$\cT(:,:,2,1,1) = \bbm  0.1445 - 0.1447i  &-0.1383 - 0.0326i\\
				-0.1370 - 0.1298i &  0.2077 + 0.1496i\ebm$} \\
			& &  \\
			&\multirow{2}{*}{\tiny	 $\cT(:,:,1,2,1) = \bbm -0.0253 - 0.0041i & -0.0781 - 0.0165i\\
				-0.0672 + 0.1327i & -0.0567 - 0.0179i \ebm$} & 	\multirow{2}{*}{\tiny	$\cT(:,:,2,2,1) = \bbm  -0.0449 - 0.0207i & -0.1598 + 0.1687i\\
				0.0385 - 0.0369i & -0.1941 - 0.2077i \ebm$} \\
			{\tiny $d = 5$} & & \\ 
			{\tiny $\F =\C$} & \multirow{2}{*}{\tiny $\cT(:,:,1,1,2) = \bbm   -0.0174 - 0.0851i &  0.1825 + 0.0812i\\
				0.1554 - 0.1894i  &-0.0990 + 0.0631i\ebm$} & \multirow{2}{*}{\tiny $\cT(:,:,2,1,2) = \bbm -0.1427 + 0.2029i & -0.1101 - 0.0420i\\
				0.1567 + 0.0221i  & 0.0613 - 0.1266i\ebm$ }\\
			& & \\ 
			& \multirow{2}{*}{\tiny $\cT(:,:,1,2,2) = \bbm   0.1962 + 0.0526i  & 0.1556 - 0.0521i\\
				0.0693 + 0.0981i  &-0.2059 - 0.2059i\ebm$}& \multirow{2}{*}{\tiny $\cT(:,:,2,2,2) = \bbm -0.1525 - 0.0337i & -0.0293 + 0.1235i\\
				0.1339 + 0.1066i &  0.1640 + 0.1765i  \ebm$} \\ 
			& & \\  \hline 
			& \multirow{2}{*}{\tiny $\cT(:,:,1,1,1,1) = \bbm   	-0.0677 + 0.0973i& 0.0106 + 0.1409i\\
				-0.0529 - 0.1296i&  -0.0443 - 0.0007i \ebm$} & \multirow{2}{*}{\tiny $\cT(:,:,2,1,1,1) = \bbm  0.1362 - 0.0385i& -0.1121 + 0.1283i\\
				-0.0286 - 0.1434i & -0.0997 - 0.1376i
			 \ebm$} \\
			& & \\  
			&\multirow{2}{*}{\tiny  $\cT(:,:,1,2,1,1) = \bbm -0.1131 + 0.0779i& -0.1089 - 0.1387i\\
				0.0179 + 0.0332i & -0.0174 - 0.0142i    \ebm$} & \multirow{2}{*}{\tiny $\cT(:,:,2,2,1,1) = \bbm  
			  -0.1153 + 0.0110i &  0.1405 - 0.1124i\\
			  0.1185 - 0.0366i&   0.0276 + 0.0703i\ebm$}\\
			& & \\ 
			& \multirow{2}{*}{\tiny  $\cT(:,:,1,1,2,1) = \bbm  -0.0674 - 0.1423i& -0.1374 - 0.0612i\\
				0.0913 - 0.0165i & -0.0321 - 0.0973i  \ebm$} & \multirow{2}{*}{\tiny $\cT(:,:,2,1,2,1) = \bbm-0.0762 + 0.1214i & 0.0930 + 0.0843i\\
				0.0724 - 0.0194i & -0.1188 - 0.0075i   
			   \ebm$}\\
			& & \\ 
			 &\multirow{2}{*}{\tiny  $\cT(:,:,1,2,2,1) = \bbm -0.0686 + 0.0616i &-0.1117 + 0.0676i\\
				-0.0427 - 0.1427i & -0.0279 + 0.0872i  \ebm$} & \multirow{2}{*}{\tiny $\cT(:,:,2,2,2,1) = \bbm  0.1469 + 0.1238i& 0.0309 + 0.0062i\\
				0.0862 + 0.0889i & 0.0763 - 0.0204i  
				\ebm$} \\ 
		{\tiny $d = 6$} 	& & \\
		{\tiny	$\F =\C$}  & \multirow{2}{*}{\tiny  $\cT(:,:,1,1,1,2) = \bbm  -0.0805 + 0.0289i & 0.1351 - 0.0362i\\
				-0.0895 - 0.0121i &  0.0901 + 0.1108i \ebm$} & \multirow{2}{*}{\tiny $\cT(:,:,2,1,1,2) = \bbm -0.0789 - 0.1379i& -0.0783 - 0.1164i\\
				0.0894 + 0.1324i & -0.0665 + 0.1341i    
				\ebm$ }\\ 
			& & \\
			& \multirow{2}{*}{\tiny  $\cT(:,:,1,2,1,2) = \bbm  0.0104 - 0.0051i &  0.1171 + 0.1257i\\
				0.0583 - 0.0200i & -0.0433 + 0.0583i \ebm$} &\multirow{2}{*}{\tiny $ \cT(:,:,2,2,1,2) = \bbm    
				0.1108 + 0.1394i&  0.1228 - 0.0153i\\
				-0.0786 + 0.1070i & 0.0241 + 0.1131i\ebm$  }\\
			& & \\
			& \multirow{2}{*}{\tiny  $\cT(:,:,1,1,2,2) = \bbm-0.1383 - 0.0247i& -0.1055 - 0.0455i\\
				-0.1247 - 0.0841i&  0.0665 + 0.1070i    \ebm$} & \multirow{2}{*}{\tiny $\cT(:,:,2,1,2,2) = \bbm  0.0313 - 0.0613i&  0.0587 + 0.1184i\\
				0.0122 + 0.0157i& -0.0778 - 0.0855i  
				\ebm$}\\ 
			& & \\  
			& \multirow{2}{*}{\tiny  $\cT(:,:,1,2,2,2) = \bbm  0.0083 + 0.0024i & -0.1214 + 0.0408i\\
				-0.0278 + 0.1363i& -0.1153 + 0.0933i  \ebm$} & \multirow{2}{*}{\tiny $ \cT(:,:,2,2,2,2) =\bbm -0.1024 + 0.1022i&  -0.0143 + 0.0358i\\
				0.0929 - 0.0538i & -0.0657 - 0.1403i   
				\ebm$} \\  
			& & \\	\hline
		\end{tabular}
	\end{scriptsize}\caption{The most entangled complex states for 500 nonsymmetric examples} \label{tensor:3:random:nonsym:qubit}
\end{table*} 

\subsubsection{Random nonsymmetric tensors}

We explore the performance of Algorithms \ref{alg:nonsym} on calculating the nuclear norm for randomly generated nonsymmetric tensors $\cT\in \F^{n\times \cdots\times n}$. The computational results are shown in Table \ref{tensor:nuclear:random:nonsym}. For each $(n,d)$ pair, we randomly generate $K_{\text{rand}}$ tensors. For each tensor, we run 30 times Algorithm \ref{alg:nonsym} with random initial points, and we choose the smallest objective value as the nuclear norm. 
The maximal (resp. minimal, average) iteration and time are calculated over these $30K_{\text{rand}}$ round calculation of Algorithm \ref{alg:nonsym}. From Table \ref{tensor:nuclear:random:nonsym}, we can see, for small examples, Algorithm \ref{alg:nonsym} can find the nuclear norm in few seconds. For relatively large examples, Algorithm \ref{alg:nonsym} needs longer time to solve.

 
	\begin{table}[htb]
	\centering
	\begin{scriptsize}
		\begin{tabular}{|c|c|c|c|c|c|c|c|c|c|} \hline
			$n$ & $d$ & $\F$ & $K_{\text{rand}}$  & $\text{Min}_{\F,\text{Iter}}$  & $\text{Avg}_{\F,\text{Iter}}$ & $\text{Max}_{\F,\text{Iter}}$ & $\text{Min}_{\F,\text{Time}}$& $\text{Avg}_{\F,\text{Time}}$  
			&   $\text{Max}_{\F,\text{Time}}$     
			\\
			\hline 
			3& 3	& $\R$ & 100 & 3 & 6.0 & 14 & 0:00:03 & 0:00:12& 0:00:30\\ \hline  
			3& 3	& $\C$ & 100 &  4 & 7.5 &13 & 0:00:12 & 0:00:25 &0:00:44\\ \hline 
			4& 3	& $\R$ & 50 & 6 & 9.3 & 17 & 0:00:13 & 0:00:28 & 0:00:53  \\ \hline
			4& 3	& $\C$ & 50& 8 & 10.3 & 15 & 0:01:25 & 0:02:44  & 0:04:19 \\ \hline
			5 & 3	& $\R$ & 50 & 7 & 11.7 &19 & 0:00:44 & 0:01:16 & 0:02:05 \\ \hline 
			3& 4	& $\R$ & 50 & 10&14.8 & 25 &0:00:54 & 0:01:51 & 0:03:29\\ \hline
			3& 4	& $\C$ & 50 & 9 & 14.1 & 23 & 0:03:50   & 0:07:45   & 0:13:28  \\ \hline  
			6 & 3	& $\R$ & 30 & 12 & 14.0	& 17 & 0:03:50  &  0:04:27  &0:05:01   \\ \hline
			5 & 3	& $\C$ & 20 & 10 & 13.3 & 19 & 0:14:32 &
			    0:27:30 & 0:41:25  \\ \hline  
			4 & 4	& $\R$ & 20 &  14 &21.4& 33 & 0:05:55  &0:12:52  & 0:21:07 \\ \hline    
			7 & 3	& $\R$ & 20 & 	15& 17.5 & 20 & 0:21:28 & 0:25:01 & 0:28:48  \\ \hline   
			8 & 3	& $\R$ & 20 & 15 & 18.0 &  23 & 1:12:50 & 1:27:26 & 1:51:58 \\ \hline  
			6 & 3	& $\C$ & 10& 14 & 16.2&	 19 & 2:43:17 & 3:17:53 & 3:55:35 \\ \hline 
			5 & 4	& $\R$ & 10 &  21 & 30.0 &  38 & 2:26:45  & 3:30:25 & 4:26:19 \\ \hline   
			4 & 4	& $\C$ & 10&  19 & 22.1	&  31 & 4:33:43 & 5:09:20 & 7:29:16 \\ \hline
		\end{tabular}\caption{Computational results for random nonsymmetric tensors} \label{tensor:nuclear:random:nonsym}
	\end{scriptsize}
\end{table}   
 	  
\subsection{Symmetric Tensors} \label{subsec:sym}
In this subsection, we test Algorithm \ref{alg:sym} on symmetric tensors. Semidefinite relaxation method \cite{Ni16} generally can find the nuclear norm for symmetric tensors, here we would like to compare the computational results of Algorithm \ref{alg:sym} with the semidefinite relaxation method. Note we only want to show Algorithm \ref{alg:sym} also works well in finding the nuclear norm. However, it might not be as fast as semidefinite relaxation method, since we need to solve several SOCP problems before it converges. We will test Algorithm \ref{alg:sym} on the following tensors: explicit symmetric $d$-qubits found from references \cite{AMM10,FLn16} and random symmetric $d$-qubits.

\subsubsection{Symmetric $d$-qubits}
We test Algorithm \ref{alg:sym} on several symmetric $d$-qubits found from references \cite{AMM10,FLn16}. The symmetric $d$-qubits are shown in Table \ref{tensor:nuclear:small:sym}. 
For convenience, we only list one element for each permutation index since the tensor is symmetric. 
Since these examples are real symmetric $d$-qubits, we calculate both real and complex nuclear norms by using two algorithms: alternating Algorithm \ref{alg:sym} and semidefinite relaxation method \cite{Ni16}. The computational results are shown in Table \ref{tensor:nuclear:small:sym}. In the table, $\|\cT\|_{1,\F}$(S) is the nuclear norm calculated by implementing Algorithm \ref{alg:sym}, and $\|\cT\|_{1,\F}$(N) is the nuclear norm calculated by semidefinite relaxation method \cite{Ni16}.
For all the symmetric $d$-qubits list in the Table \ref{tensor:nuclear:small:sym}, the semidefinite relaxation method found exact nuclear norms with certification. For the certification conditions of semidefinite relaxation method, please refer to \cite{Ni16} for details.
 From the Table \ref{tensor:nuclear:small:sym}, we can see that Algorithm \ref{alg:sym} also finds all nuclear norms with tiny numerical errors. The spectral norms of these symmetric $d$-qubits list in Table \ref{tensor:nuclear:small:sym} are shown in Table \ref{tensor:spectral:small:sym}. The computational time comparison results are shown in Table \ref{tensor:nuclear:small:sym:time} in seconds. $\text{Avg}_{\F,\text{Time}}$(S) 
 is the average computational time of Algorithm \ref{alg:sym} over field $\F$, and $\text{Avg}_{\F,\text{Iter}}$(S) is the average iteration of Algorithm \ref{alg:sym} over field $\F$. Time(N,$\F$) is the computational time of semidefinite relaxation method \cite{Ni16} over field $\F$.

\begin{table}[htb]
	\centering
	\begin{scriptsize}
		\begin{tabular}{|c|c|p{4.5cm}||c|c||c|c|c|c|}  \hline
			No.   & $d$ & \quad\quad\quad\quad \quad\quad  Tensor & \scriptsize{$\|\cT\|_{1,\R}$(S)}  & $\|\cT\|_{1,\C}$(S) & $\|\cT\|_{1,\R}$(N)  & $\|\cT\|_{1,\C}$(N)  
			\\
			\hline  1 & 3 & $\cT_{1,1,1} = \cT_{2,2,2} = \frac{1}{\sqrt{2}}$ & 1.4149 & 1.4141 & 1.4142  & 1.4142   \\ 
			\hline  2 & 3& $\cT_{1,1,2} = \frac{1}{\sqrt{3}}$ & 1.7321 & 1.5000 & 1.7321 & 1.5000 \\
			\hline  \multirow{3}{*}{3} & \multirow{3}{*}{3}& $\cT_{1,1,1} =\cT_{1,1,2}=\frac{1}{\sqrt{20}}$, $\cT_{2,2,1}= \cT_{2,2,2} = -\frac{2}{\sqrt{20}}$ & \multirow{3}{*}{1.9235}  & \multirow{3}{*}{1.4238}  & \multirow{3}{*}{1.9235} & \multirow{3}{*}{1.4230}   \\
			\hline  \multirow{2}{*}{4} & \multirow{2}{*}{3} & $\cT_{1,1,1}=0.3358,~\cT_{2,2,2} = -0.4283$, $\cT_{1,1,2}  =0.4305$, $\cT_{2,2,1}= -0.2220$ & \multirow{2}{*}{1.9903} & \multirow{2}{*}{1.4139} & \multirow{2}{*}{1.9903} & \multirow{2}{*}{1.4138}   \\
			\hline  5 & 3 & $\cT_{1,1,2} = \frac{1}{2},\cT_{2,2,2} = -\frac{1}{2}$ & 2.0000 & 1.4142 & 2.0000 & 1.4142    \\
			\hline  6 & 4 & $\cT_{1,1,1,2}  =\frac{1}{2}$ & 2.0000  & 1.5398 & 2.0000 & 1.5396  \\
	  	\hline  7 & 4 & $\cT_{1,1,1,1} = \frac{1}{\sqrt{3}},\cT_{2,2,2,1} =\frac{1}{2}\sqrt{\frac{2}{3}}$ &  1.9112  & 1.7329  & 1.9107 & 1.7321\\
			\hline 	8 & 5 & $\cT_{1,1,1,1,2}=\frac{1}{\sqrt{5}}$ & 2.2359 & 1.5629& 2.2361  & 1.5625     \\  
			\hline  \multirow{3}{*}{9} & \multirow{3}{*}{5} & $\cT_{1,1,1,1,1}=\frac{1}{\sqrt{1+A^2}},\cT_{1,2,2,2,2}=\frac{A}{\sqrt{5(1+A^2)}} $, where $A\approx 1.53154$ & \multirow{3}{*}{2.4190}  & \multirow{3}{*}{1.8291}     & \multirow{3}{*}{2.4190} & \multirow{3}{*}{1.8291} \\  
			\hline
			\multirow{2}{*}{10} & \multirow{2}{*}{6} &  $\cT_{1,1,1,1,1,2} = \frac{1}{2\sqrt{3}}$ & \multirow{2}{*}{3.4641} &\multirow{2}{*}{2.1219}& \multirow{2}{*}{3.4641} &\multirow{2}{*}{2.1213}  \\
			 & &  $\cT_{1,2,2,2,2,2} = \frac{1}{2\sqrt{3}}$ &&&&\\
			\hline   		   
		\end{tabular}
	\end{scriptsize}\caption{Computational results for symmetric $d$-qubits} \label{tensor:nuclear:small:sym}
\end{table}

\begin{table}[htb]
	\centering
	\begin{scriptsize}
		\begin{tabular}{|c|c|c|c|c|c|c|c|c|c|c|c|}  \hline
No. & 1 & 2 &3 &4 & 5&6 & 7&8&9 & 10\\ \hline
$\|\cT\|_{\infty,\R}$ & 0.7071 & 0.6667 & 0.6938 & 0.5785 & 0.5000 & 0.6495 & 0.5774 & 0.4472 &0.5492 & 0.4714 \\ \hline
$\|\cT\|_{\infty,\C}$ & 0.7071 & 0.6667 & 0.7171 & 0.7075 & 0.7071 & 0.6495 & 0.5774 & 0.6400 & 0.5492 & 0.4714\\ \hline	
{\sf P}$_{\R}$ & 1.0005 & 1.1548 & 1.3345 & 1.1514 & 1.0000 & 1.0000 & 1.1035 & 1.0000 & 1.3285 & 1.6329\\ \hline
{\sf P}$_{\C}$ & 1.0000 & 1.0000 & 1.0210 & 1.0003 & 1.0000 & 1.0001 & 1.0006 & 1.0002 & 1.0045 & 1.0003\\ \hline			   
		\end{tabular}
	\end{scriptsize}\caption{Spectral norms of symmetric $d$-qubits in Table \ref{tensor:nuclear:small:sym}} \label{tensor:spectral:small:sym}
\end{table}

\begin{table}[htb]
	\centering
	\begin{scriptsize}
		\begin{tabular}{|c||c|c||c|c||c|c|c|c|}  \hline
			No.   &   $\text{Avg}_{\R,\text{Time}}$(S) &  $\text{Avg}_{\C,\text{Time}}$(S) & $\text{Avg}_{\R,\text{Iter}}$(S) &  $\text{Avg}_{\C,\text{Iter}}$(S)  & Time(N,$\R$)& Time(N,$\C$)  \\
			\hline  1 & 2.65 & 4.09 & 3.4 & 13.1 & 0.07 & 0.15 \\ 
			\hline  2 & 1.76  & 2.22 & 1.9 & 5.9 & 0.10 & 0.13 \\
			\hline  3 & 1.51 & 3.59 & 1.5 & 10.2 & 0.11 & 0.12  \\
			\hline  4 & 1.50 & 3.95 & 1.5 & 10.9  & 0.10 & 0.11 \\
			\hline  5 & 1.58 & 2.73 & 1.5 & 6.8 & 0.07 & 0.10  \\
			\hline  6 & 1.89 & 8.74 & 2.6 & 6.5 & 0.11 & 1.11\\
			\hline 	7 & 4.32 & 6.51 & 5.8 & 6.7 &  0.13 & 1.31 \\  
			\hline 8 & 1.70 & 3.98 & 3.6 & 8.6 & 0.11 & 0.34 \\
			\hline 9 &  1.88 & 4.35  &3.6 & 8.8& 0.11 &  0.36 \\
			\hline 10 & 4.79  & 4.82  & 7.3 &  8.2& 0.19 & 0.37\\
			\hline
		\end{tabular}
	\end{scriptsize}\caption{Computational time comparison for symmetric $d$-qubits in Table \ref{tensor:nuclear:small:sym}} \label{tensor:nuclear:small:sym:time}
\end{table}

 \subsubsection{Random symmetric $d$-qubits}
Similar like nonsymmetric case, we are also interested in finding the maximum entangled states. 
Here we randomly generate 5000 symmetric states over field $\F$ for $d=3,4,5,6$, and calculate their nuclear norms over field $\F$ by using alternating Algorithm \ref{alg:sym} and semidefinite relaxation method \cite{Ni16}. Algorithm \ref{alg:sym} is implemented with random starting point for 30 times, and choose the one with smallest objective value as the nuclear norm. For all these 5000 randomly examples, two algorithms find the same nuclear norm with tiny numerical errors for Algorithm \ref{alg:sym}. So Algorithm \ref{alg:sym} also works well for calculating nuclear norms for symmetric tensors.
We show the most entangled symmetric states among the 5000 random examples in Table \ref{tensor:1:real:most:bosons} and Table \ref{tensor:2:comp:most:bosons}
for $\F = \R$ and $\F = \C$ respectively.

\begin{table*}[htb]
	\centering
	\begin{scriptsize}
		\begin{tabular}{|c|p{5.8cm}|c|c|c|c|c|c|}  \hline
			$d$ & \quad\quad\quad\quad\quad\quad \quad\quad Tensor & $\|\cT\|_{1,\R}$  &   $\|\cT\|_{\infty,\R}$  &  {\sf P}$_{\R}$\\ 
			\hline   
			\multirow{2}{*}{3} & $\cT_{1,1,1} = -0.4950, \cT_{1,2,1}= -0.1078$ & \multirow{2}{*}{1.9999}   &  \multirow{2}{*}{0.5001} &    \multirow{2}{*}{1.0000}  \\ 
			&$\cT_{1,2,2} = 0.4864, \cT_{2,2,2}= 0.1018$&&&\\
			\hline 
			\multirow{3}{*}{4} &  $\cT_{1,1,1,1} = -0.3132; \cT_{1,1,1,2} = 0.1703;$ & \multirow{3}{*}{2.8283}  & \multirow{3}{*}{0.3581} &   \multirow{3}{*}{1.0128}\\ 
			&$ \cT_{1,1,2,2} =0.3089; \cT_{1,2,2,2} = -0.1737;$&&&\\
			&$ \cT_{2,2,2,2} = -0.3042$&&&\\
			\hline
			\multirow{3}{*}{5} &  $\cT_{1,1,1,1,1} = 0.2388; \cT_{1,1,1,1,2} = 0.0952;$ & \multirow{3}{*}{3.9984}   & \multirow{3}{*}{0.2716} &   \multirow{3}{*}{1.0860}  \\ 
			&$ \cT_{1,1,1,2,2} =  -0.2304;\cT_{1,1,2,2,2} =  -0.0947; $&&&\\
			&$\cT_{1,2,2,2,2} =  0.2289; \cT_{2,2,2,2,2} = 0.1223$&&&\\
			\hline 
			\multirow{4}{*}{6} & $\cT_{1,1,1,1,1,1} = 0.0541; \cT_{1,1,1,1,1,2} = 0.1972;  $ & \multirow{4}{*}{5.5931}  & \multirow{4}{*}{0.2049}  & \multirow{4}{*}{1.1464} \\
			&$\cT_{1,1,1,1,2,2} = -0.0598; \cT_{1,1,1,2,2,2} =-0.170;  $&&&\\
			&$\cT_{1,1,2,2,2,2} = 0.0273; \cT_{1,2,2,2,2,2} = 0.1411;$&&&\\
			&$\cT_{2,2,2,2,2,2} = 0.0274$&&&\\
			\hline     				
		\end{tabular}
	\end{scriptsize}\caption{The most entangled real symmetric states for 5000 random examples.} \label{tensor:1:real:most:bosons}
\end{table*}

\begin{table*}[htb]
	\centering
	\begin{scriptsize}
		\begin{tabular}{|c|p{8.2cm}|c|c|c|c|c|c|}  \hline
			$d$ & \quad\quad\quad\quad\quad\quad \quad\quad Tensor  & $\|\cT\|_{1,\C}$ &    $\|\cT\|_{\infty,\C}$ &  {\sf P}$_{\C}$\\
			\hline   
			\multirow{2}{*}{3} & $\cT_{1,1,1} = 0.3029 - 0.3436i, \cT_{1,2,1}= -0.3423 - 0.2033i$ & \multirow{2}{*}{1.4985} &  \multirow{2}{*}{0.6869}  &  \multirow{2}{*}{1.0293}    \\ 
			&$\cT_{1,2,2} = 0.0727 - 0.3054i, \cT_{2,2,2}= 0.1365 + 0.0183i$&&&\\
			\hline   
			\multirow{3}{*}{4} &  $\cT_{1,1,1,1} = 0.2679-0.2422i, \cT_{1,1,1,2} = 0.3086+0.0143i$ & \multirow{3}{*}{1.7247}& \multirow{3}{*}{0.6136} & \multirow{3}{*}{1.0583}  \\ 
			&$ \cT_{1,1,2,2} =-0.1385 -0.1266i , \cT_{1,2,2,2} = -0.0986-0.0715i $&&&\\
			&$ \cT_{2,2,2,2} =  0.1396 - 0.4449i$&&&\\
			\hline 
			\multirow{3}{*}{5} &  $\cT_{1,1,1,1,1} = 0.2025-0.1845i , \cT_{1,1,1,1,2} =  0.1868-0.2069i$ & \multirow{3}{*}{1.8112} &\multirow{3}{*}{0.6236 } & \multirow{3}{*}{1.1295}   \\ 
			&$ \cT_{1,1,1,2,2} = 0.0060 +0.1202i,\cT_{1,1,2,2,2} = 0.0177-0.0913i $&&&\\
			&$\cT_{1,2,2,2,2} = -0.0953-0.1989i, \cT_{2,2,2,2,2} =  0.1904-0.1606i$&&&\\
			\hline  
			\multirow{4}{*}{6} & $\cT_{1,1,1,1,1,1} = 0.0907 -   0.0477i, \cT_{1,1,1,1,1,2} = -0.1934-0.2027i  $ & \multirow{4}{*}{2.0312} & \multirow{4}{*}{0.6271}  & \multirow{4}{*}{1.2738}   \\
			&$\cT_{1,1,1,1,2,2} = 0.0655+0.0558i, \cT_{1,1,1,2,2,2} =0.0196+0.0043i  $&&&\\
			&$\cT_{1,1,2,2,2,2} = 0.0371 +   0.0115i, \cT_{1,2,2,2,2,2} =  0.1736 - 0.1746i $&&&\\
			&$\cT_{2,2,2,2,2,2} = -0.0595 + 0.0973i$&&&\\
			\hline     				
		\end{tabular}
	\end{scriptsize}\caption{The most entangled complex symmetric states for 5000 random examples.} \label{tensor:2:comp:most:bosons}
\end{table*} 

 \clearpage
\bibliographystyle{plain}

\begin{thebibliography}{MMM}
 \bibitem{AA13} S. Aaronson and A. Arkhipov, The computational complexity of linear optics, \emph{Theory of Computing} 9 (2013), 143--252.
\bibitem{AMM10} M. Aulbach, D. Markham and Mi. Murao,The maximally entangled symmetric state in terms of the geometric measure, \emph{New Journal of Physics}, 12 (2010), pp. 073025. 	
 	
 \bibitem{ACCO} K.D. Andersen, E. Christiansen, A.R. Conn and M.L. Overton, 
An Efficient Primal-Dual Interior-Point Method for Minimizing a Sum of Euclidean Norms, \emph{SIAM J. Scient. Comp.} 22 (2000), pp. 243-262. 
\bibitem{Ban38} S. Banach, \"Uber homogene Polynome in ($L^2$),' \emph{Studia Math.}, \textbf{7} (1938), pp.~36--44.
\bibitem{BHSW06}
D.J. Bates, J. D. Hauenstein, A.J. Sommese,
and C.W. Wampler.
\newblock Bertini: Software for Numerical Algebraic Geometry,
\newblock Available at bertini.nd.edu with permanent doi: dx.doi.org/10.7274/R0H41PB5. 
 
\bibitem{BV04}  S. Boyd and L. Vandenberghe, \emph{Convex Optimization}, Cambridge University Press, 2004, ISBN 978-0-521-83378-3.
 
\bibitem{Stephen2007} S. Brierley, A. Higuchi , On maximal entanglement between two pairs in four-qubit pure states,  \emph{Journal of Physics A: Mathematical and Theoretical} 40(29) (2007), 8455.

\bibitem{CS} D. Cartwright, B. Sturmfels, The number of eigenvectors of a tensor,  \emph{Linear Algebra Appl.} 438 (2013), no. 2, 942-–952.
\bibitem{CHLZ12} B. Chen, S. He, Z. Li, and S. Zhang, Maximum block improvement and polynomial optimization,
 \emph{SIAM J. OPTIM.} 22 (2012),  87--107.
 \bibitem{CXZ10} L. Chen, A. Xu and H. Zhu, Computation of the geometric measure of entanglement for pure multiqubit states,  \emph{Physical Review A} 82, 032301, 2010.
\bibitem{LMV00} L. de Lathauwer, B. de Moor, and J. Vandewalle, On the best rank-1 and
rank-$(R_1,R_2, . . . ,R_N)$ approximation of higher-order tensors, \emph{SIAM J. Matrix Anal.
Appl.}, 21 (2000), pp. 1324--1342.
\bibitem{DVC00} W. D\"ur, G. Vidal and J.I. Cirac, Three qubits can be entangled in two inequivalent ways, \emph{Phys. Rev. A.} 62 (2000), 062314.
\bibitem{EPR35} A. Einstein, B. Podolsky, and N. Rosen N, Can Quantum-Mechanical Description of Physical Reality Be Considered Complete?,  \emph{Phys. Rev.} 47, is. 10, (1935),  777--780.
 \bibitem{ES09} L. Eld\'en and B. Savas, A newton-grassmann method for computing the best multilinear rank-$(r_1, r_2, r_3)$ approximation of a tensor,
\emph{SIAM Journal on Matrix Analysis and applications}, 31 (2009), 248--271.
\bibitem{Fan57} U. Fano, Description of States in Quantum Mechanics by Density Matrix and Operator Techniques, \emph{Reviews of Modern Physics} 29 (1957), 74--93.
 \bibitem{Fri13} S. Friedland, Best rank one approximation of real symmetric tensors can be chosen symmetric, 
 \textit{Front. Math. China}, 8 (1) (2013), 19--40.
\bibitem{FK16} S. Friedland and T. Kemp, Most Boson quantum states are almost maximally entangled, arXiv:1612.00578.
 \bibitem{FL14} S. Friedland and L.-H. Lim, Computational Complexity of Tensor Nuclear Norm,  arXiv:1410.6072v1.
\bibitem{FLd16}  S. Friedland and L.-H. Lim, The computational complexity of duality,  \emph{SIAM Journal on Optimization}, 26, No. 4 (2016), 2378--2393
\bibitem{FLn16} S. Friedland and L.-H. Lim, Nuclear norm of higher-order tensors, to appear in \emph{Mathematics of Computation}, arXiv:1410.6072v3.
\bibitem{FMPS} S. Friedland,  V. Mehrmann, R. Pajarola and S.K. Suter, On best rank one approximation of tensors, \emph{ Numerical Linear Algebra with Applications}, 20 (2013), 942--955.
\bibitem{FO14} S. Friedland and G. Ottaviani, The number of singular vector tuples and uniqueness of best rank one approximation of tensors, 
 \emph{Foundations of Computational Mathematics} 14, 6 (2014), 1209--1242.
\bibitem{FT15}S. Friedland and V. Tammali, Low-rank approximation of tensors, \emph{Numerical Algebra, Matrix Theory, Differential-Algebraic Equations and 
Control Theory}, edited by P. Benner et all, Springer, 2015, 377-410.
\bibitem{FW16}  S. Friedland and L. Wang, Geometric measure of entanglement of symmetric d-qubits is polynomial-time computable,  arXiv:1608.01354.
\bibitem{Ga08} S. Gharibian, Strong NP-hardness of the Quantum separability problem, \textit{Quantum Inf.\ Comput.}, \textbf{10} (2010), no.~3--4,  pp.~343--360.
\bibitem{GFE09} D. Gross, S. T. Flammia, and J. Eisert, Most Quantum States Are Too Entangled To Be Useful As Computational Resources, \emph{Phys. Rev. Lett.} 102, 190501, 2009.
\bibitem{Gu02}  L.~Gurvits, Classical deterministic complexity of Edmonds problem and quantum entanglement,
\textit{Proc.\ ACM Symp.\ Theory Comput.} (STOC), \textbf{35}, pp.~10--19, New York, NY, ACM Press, 2003.

\bibitem{henrion2009gloptipoly} D. Henrion, J.-B. Lasserre and J. L\"{o}fberg. GloptiPoly 3: moments, optimization and semidefinite programming. {\em Optimization Methods \& Software}, 24 (2009), pp. 761--779.
\bibitem{HS00} A.Higuchi, A. Sudbery, How entangled can two couples get? \emph{Physics Letters A}, 273(4) (2000), pp. 213-217. 

\bibitem{HL13} C.J. Hillar and L.-H. Lim, Most tensor problems are NP-hard, \emph{J.\ Assoc.\ Comput.\ Mach.}, 60 (2013), no.~6, p.~45.

\bibitem{Hor96} M. Horodecki, P. Horodecki, and R. Horodecki, Separability of mixed states: necessary and sufficient conditions, \emph{J. Physics Letters A.} 223 (1996), 1--8.

\bibitem{Hubetall09} R. H\"ubener, M. Kleinmann, T.-C. Wei, C. Gonz\'alez-Guill\'en, and O. G\"uhne, \emph{Phys. Rev. A} 80 (2009), 032324. 
 \bibitem{KB09} T.G. Kolda and B.W. Bader, Tensor decompositions and applications, \emph{SIAM Review}
 51 (2009), pp. 455--500.
 
  \bibitem{yalmip} J.L\"{o}fberg, YALMIP: A toolbox for modeling and optimization in MATLAB,  \emph{In Computer Aided Control Systems Design, 2004 IEEE International Symposium on IEEE}, (2005), pp. 284-289. 
  \bibitem{Nak08} M. Nakahara and T. Ohmi, \emph{Quantum Computing: from Linear Algebra to Physical Realizations}, CRC Press, 2008.
 
\bibitem{Ni16} J. Nie, Symmetric Tensor Nuclear Norms, 	arXiv:1605.08823.
\bibitem{NW14} J. Nie and L. Wang, Semidefinite Relaxations for Best Rank-1 Tensor Approximations, \emph{SIAM Journal on Matrix Analysis and Applications},  35 (2014), no. 3, 1155-1179. 

\bibitem{NW12} J. Nie and L. Wang, Regularization methods for SDP relaxations in large-scale polynomial optimization, \emph{SIAM Journal on Optimization}, 22(2012), no. 2, 408-428. 
\bibitem{NZ16} J. Nie and X. Zhang, Positive Maps and Separable Matrices, \emph{SIAM Journal on Optimization}, 26 (2016), No. 2, 1236--1256.
\bibitem{Per96} A. Peres, Separability Criterion for Density Matrices, \emph{Phys. Rev. Lett.} 77 (1996), 1413--1415.
\bibitem{Ru02} O. Rudolph, Further results on the cross norm criterion for separability, \emph{Quantum Information Processing}, 4 (2005), Issue 3, 219--239.
\bibitem{SL10} B. Savas and L.-H. Lim, Quasi-Newton methods on Grassmannians and multilinear approximations of tensors,
\emph{SIAM Journal on Scientific Computing} 32 (2010),  3352--3393.
\bibitem{Sch35} E. Schr\"{o}dinger, Discussion of probability relations between separated systems, \emph{Mathematical Proceedings of the Cambridge Philosophical Society}, 31 is. 4, (1935),  555--563. 
\bibitem{Sch36} E. Schr\"{o}dinger, Probability relations between separated systems,  \emph{Mathematical Proceedings of the Cambridge Philosophical Society}, 32,  is. 3, (1936), 446--452.
\bibitem{sturm1999using} J.F. Sturm. Using SeDuMi 1.02, a MATLAB toolbox for optimization over symmetric cones. {\em Optimization methods and software}, 11 (1999), pp. 625--653.
\bibitem{TWP09} S. Tamaryan, T.-C. Wei, D. Park, Maximally entangled three-qubit states via geometric measure of entanglement, \emph{Phys. Rev.} A 80 (2009), 052315.
\bibitem{TH00} B.M. Terhal and P. Horodecki, A Schmidt number for density matrices, \emph{Phys. Rev}. A 61 (2000), 040301.

\bibitem{tutuncu2001sdpt3} R.H. T{\"u}t{\"u}nc{\"u}, K.C. Toh and M.J. Todd, SDPT3 -- a Matlab software package for semidefinite-quadratic-linear programming, version 3.0. Web page: {http://www. math. nus. edu. sg/\~{} mattohkc/sdpt3. html}, 2001.

 \bibitem{TTT06}  R.H. Tutuncu,  K.C. Toh and  M.J. Todd, Solving semidefinite-quadratic-linear programs using SDPT3,  \emph{Mathematical Programming} B 95 (2003), 189--217.
 
 \bibitem{WG03} T.-C. Wei and P.M. Goldbart, Geometric measure of entanglement and applications to bipartite and multipartite quantum states, \emph{Phys. Rev.} A 68 (2003), 042307.
\bibitem{ZhaG01} T. Zhang and G.H. Golub. Rank-one approximation to high order tensors.
{\em SIAM J. Matrix Anal. Appl}. 23 (2001), 534--550.






 
 \end{thebibliography}

\end{document}